\documentclass[12pt]{article}

\usepackage[authoryear,round]{natbib}

\usepackage{bbm}
\usepackage{amsmath} 
\usepackage{amssymb}
\usepackage{amsthm}
\usepackage{authblk}
\usepackage{url}
\usepackage{subcaption}
\usepackage{tikz}
\usepackage[margin=1in]{geometry}
\usepackage[colorlinks,citecolor=blue,linkcolor=black,urlcolor=blue]{hyperref}

\usepackage{environ}

\newtheorem{theorem}{Theorem}
\newtheorem{proposition}{Proposition}
\newtheorem{lemma}{Lemma}

\newtheorem{condition}{Condition}
\newtheorem{remark}{Remark}

\newcommand{\mytrans}{\top}

\newcommand{\vzero}{\boldsymbol{\mathbf{0}}} 
\newcommand\reals{\mathbb{R}} 

\newcommand*\diffF{\mathop{}\!\mathrm{d} F} 

\newcommand{\argmin}{\mathop\mathrm{arg min}}

\newcommand{\indifun}{\mathbbm{1}} 
\newcommand{\sgn}{\mathop\mathrm{sign}} 
\newcommand{\softthre}{\mathcal{T}} 
\newcommand{\Tr}{\mathop\mathrm{tr}} 
\newcommand{\Fnorm}[1]{\| #1 \|_{\text{F}}} 

\newcommand{\Prob}{\Pr} 
\newcommand{\E}{\mathbb{E}} 
\newcommand{\ind}{\mathrel{\perp\!\!\!\perp}} 

\newcommand{\toP}{\rightarrow_p} 
\newcommand{\tod}{\rightarrow_d} 

\newcommand{\Gsn}{\mathcal{N}}

\newcommand{\vect}[1]{\boldsymbol{#1}} 

\newcommand{\vectrv}[1]{\boldsymbol{#1}}

\newcommand{\vT}{\boldsymbol{\mathbf{T}}} 
\newcommand{\vM}{\boldsymbol{\mathbf{M}}} 
\newcommand{\vY}{\boldsymbol{\mathbf{Y}}} 
\newcommand{\vE}{\boldsymbol{\mathbf{E}}} 
\newcommand{\vepsilon}{\boldsymbol{\mathbf{\epsilon}}} 

\newcommand{\vu}{\boldsymbol{\mathbf{u}}} 

\newcommand{\vunhatA}{\hat{\vu}_n^{\alpha}} 
\newcommand{\vunhatB}{\hat{\vu}_n^{\beta}} 

\newcommand{\AS}[1][]{\mathcal{A}^{\ast #1}} 
\newcommand{\ASA}[1][]{\mathcal{A}_{\alpha}^{\ast #1}} 
\newcommand{\ASB}[1][]{\mathcal{A}_{\beta}^{\ast #1}} 
\newcommand{\AShat}{\hat{\mathcal{A}}} 
\newcommand{\ASBhat}{\hat{\mathcal{A}}_{n,\beta}}

\newcommand{\BT}{\eta} 
\newcommand{\BTn}{\BT^{\ast}} 
\newcommand{\An}[1][]{\boldsymbol{\mathbf{\alpha}}^{\ast #1}} 
\newcommand{\Bn}[1][]{\boldsymbol{\mathbf{\beta}}^{\ast #1}} 
\newcommand{\AnASA}[1][]{\boldsymbol{\mathbf{\alpha}}_{ \ASA}^{\ast #1}}

\newcommand{\BnASB}[1][]{\boldsymbol{\mathbf{\beta}}_{ \ASB}^{\ast #1}}
\newcommand{\Anj}[1][j]{\alpha_{#1}^{\ast}} 
\newcommand{\Bnj}[1][j]{\beta_{#1}^{\ast}}

\newcommand{\Bnp}[1][]{\boldsymbol{\mathbf{\beta}}_{n}^{\ast '}} 

\newcommand{\Anp}[1][]{\boldsymbol{\mathbf{\alpha}}_{n}^{\ast '}} 

\newcommand{\Anhat}[1][]{\hat{\boldsymbol{\mathbf{\alpha}}}_{n #1}}
\newcommand{\Bnhat}[1][]{\hat{\boldsymbol{\mathbf{\beta}}}_{n #1}}
\newcommand{\Anjhat}[1][j]{\hat{\alpha}_{n #1}}
\newcommand{\Bnjhat}[1][j]{\hat{\beta}_{n #1}}

\newcommand{\AnhatInit}{\hat{\boldsymbol{\mathbf{\alpha}}}_n^0}
\newcommand{\BnhatInit}{\hat{\boldsymbol{\mathbf{\beta}}}_n^0}
\newcommand{\AnjhatInit}[1][j]{\hat{\alpha}_{n #1}^0}
\newcommand{\BnjhatInit}[1][j]{\hat{\beta}_{n #1}^0}

\newcommand{\BnOLS}{\hat{\boldsymbol{\mathbf{\beta}}}_{n,o}}

\newcommand{\AnjOLS}{\hat{\alpha}_{nj,o}}
\newcommand{\BnjOLS}{\hat{\beta}_{nj,o}}

\newcommand{\SigmaE}{\vect{\Sigma}} 

\newcommand{\SigmaX}{\vect{\Sigma}_{X}}

\newcommand{\lambdaA}{\lambda_{n \alpha}}
\newcommand{\lambdaB}{\lambda_{n \beta}}
\newcommand{\gAOne}{\gamma_{\alpha}}
\newcommand{\gATwo}{\eta_{\alpha}}
\newcommand{\gBOne}{\gamma_{\beta}}
\newcommand{\gBTwo}{\eta_{\beta}}

\newcommand{\wA}[1][j]{\hat{w}_{n #1, \alpha}}
\newcommand{\wB}[1][j]{\hat{w}_{n #1, \beta}}

\newcommand{\C}{C} 
\newcommand{\bSim}{\delta} 
\newcommand{\TrnjA}{l_{nj,\alpha}} 
\newcommand{\TrnjB}{l_{nj,\beta}} 
\newcommand{\CTr}{l_0} 

\newcommand{\valpha}{\boldsymbol{\mathbf{\alpha}}}
\newcommand{\vbeta}{\boldsymbol{\mathbf{\beta}}}

\newcommand{\VnA}{V_n^{\alpha}} 
\newcommand{\VA}{V^{\alpha}} 
\newcommand{\VnB}{V_n^{\beta}} 
\newcommand{\VB}{V^{\beta}} 

\newcommand{\newm}{JAP} 
\newcommand{\AL}{AL}

\newcommand{\DM}{\mathbf{D}_{M}}
\newcommand{\PDM}{\mathbf{P}_{\mathbf{D}_{M}}}
\newcommand{\PDMP}{\mathbf{P}_{\mathbf{D}_{M}}^{\perp}}

\newcommand{\X}{\mathbf{X}_{n}}

\newcommand{\PXP}{\mathbf{P}_{\mathbf{X}_{n}}^{\perp}}

\newcommand{\CA}{\sigma^2_T} 

\newcommand{\thinLineWidth}{0.08mm} 
\newcommand{\thickLineWidth}{0.4mm} 
\newcommand{\ultraThickLineWidth}{0.8mm} 
\newcommand{\tikzsize}{1.5}

\newcommand{\FigOnesingle}{\tikz[scale=\tikzsize]{
    \node (T) at (0,0) {$T$};
    \node (M) at (1.2,0) {$M_j$};
    \node (Y) at (2.4,0) {$Y$};
    \draw[->,line width=\thickLineWidth] (T) -- node[pos=0.5, above] {$\alpha_j^\ast$} (M);
    \draw[->,line width=\thickLineWidth] (M) -- node[pos=0.5, above] {$\beta_j^\ast$} (Y);
}}

\newcommand{\FigTwosingle}{\tikz[scale=\tikzsize]{
    \node (T) at (0,0) {$T$};
    \node (M) at (1.2,0) {$M_j$};
    \node (Y) at (2.4,0) {$Y$};
    \draw[->,line width=\ultraThickLineWidth] (T) -- node[pos=0.5, above] {$\alpha_j^\ast$} (M);
    \draw[->,line width=\thinLineWidth] (M) -- node[pos=0.5, above] {$\beta_j^\ast$} (Y);
}}

\newcommand{\FigThreesingle}{\tikz[scale=\tikzsize]{
    \node (T) at (0,0) {$T$};
    \node (M) at (1.2,0) {$M_j$};
    \node (Y) at (2.4,0) {$Y$};
    \draw[->,line width=\thinLineWidth] (T) --  node[pos=0.5, above] {$\alpha_j^\ast$}(M);
    \draw[->,,line width=\ultraThickLineWidth] (M) -- node[pos=0.5, above] {$\beta_j^\ast$} (Y);
}}

\newcommand{\FigOne}{\tikz[scale=\tikzsize]{
    \node (T) at (0,0) {$T$};
    \node (M) at (1.5,0) {$M_j$};
    \node (Y) at (3,0) {$Y$};
    \draw[->,line width=\thickLineWidth] (T) -- node[pos=0.5, above] { $\alpha_j^\ast$} (M);
    \draw[->,line width=\thickLineWidth] (M) -- node[pos=0.5, above] { $\beta_j^\ast$} (Y);
    \draw[->,line width=\thickLineWidth,dashed] (T) .. controls (0.2, -0.8) and (2.8, - 0.8) ..      node[pos=0.5, above] {$\eta^\ast$} (Y); 
}}

\newcommand{\FigTwo}{\tikz[scale=\tikzsize]{
    \node (T) at (0,0) {$T$};
    \node (M) at (1.5,0) {$M_j$};
    \node (Y) at (3,0) {$Y$};
    \draw[->,line width=\ultraThickLineWidth] (T) -- node[pos=0.5, above] { $\alpha_j^\ast$} (M);
    \draw[->,line width=\thinLineWidth] (M) -- node[pos=0.5, above] { $\beta_j^\ast$} (Y);
    \draw[->,line width=\thickLineWidth,dashed] (T) .. controls (0.2, -0.8) and (2.8, - 0.8) ..      node[pos=0.5, above] { $\eta^\ast$} (Y); 
}}
\newcommand{\FigThree}{\tikz[scale=\tikzsize]{
    \node (T) at (0,0) {$T$};
    \node (M) at (1.5,0) {$M_j$};
    \node (Y) at (3,0) {$Y$};
    \draw[->,line width=\thinLineWidth] (T) --  node[pos=0.5, above] { $\alpha_j^\ast$}(M);
    \draw[->,,line width=\ultraThickLineWidth] (M) -- node[pos=0.5, above] { $\beta_j^\ast$} (Y);
    \draw[->,line width=\thickLineWidth,dashed] (T) .. controls (0.2, -0.8) and (2.8, - 0.8) ..      node[pos=0.5, above] { $\eta^\ast$} (Y);
}}
\newcommand{\FigFour}{\tikz[scale=\tikzsize]{
    \node (T) at (0,0) {$T$};
    \node (M) at (1.5,0) {$M_j$};
    \node (Y) at (3,0) {$Y$};
    \node at (0.8, 0.25) {\textcolor{gray}{$\alpha_j^{\ast} = 0$}};
    \draw[->,line width=\thickLineWidth] (M) -- node[pos=0.5, above] { $\beta_j^\ast$} (Y);
    \draw[->,line width=\thickLineWidth,dashed] (T) .. controls (0.2, -0.8) and (2.8, - 0.8) ..      node[pos=0.5, above] { $\eta^\ast$} (Y);
}}
\newcommand{\FigFive}{\tikz[scale=\tikzsize]{
    \node (T) at (0,0) {$T$};
    \node (M) at (1.5,0) {$M_j$};
    \node (Y) at (3,0) {$Y$};
    \draw[->,line width=\thickLineWidth] (T) -- node[pos=0.5, above] {$\alpha_j^\ast$} (M);
    \node at (2.3, 0.25) {\textcolor{gray}{$\beta_j^\ast = 0$}};
    \draw[->,line width=\thickLineWidth,dashed] (T) .. controls (0.2, -0.8) and (2.8, - 0.8) ..      node[pos=0.5, above] { $\eta^\ast$} (Y);
}}
\newcommand{\FigSix}{\tikz[scale=\tikzsize]{
    \node (T) at (0,0) {$T$};
    \node (M) at (1.5,0) {$M_j$};
    \node (Y) at (3,0) {$Y$};
    \node at (0.8, 0.25) {\textcolor{gray}{$\alpha_j^{\ast} = 0$}};
    \node at (2.3, 0.25) {\textcolor{gray}{$\beta_j^\ast = 0$}};
    \draw[->,line width=\thickLineWidth,dashed] (T) .. controls (0.2, -0.8) and (2.8, - 0.8) ..      node[pos=0.5, above] { $\eta^\ast$} (Y);
}}

\newcommand{\JointMA}{
    \begin{tikzpicture}[node distance=3cm and 3cm, >=stealth]
        \node (T) at (0,0) {$T$};
        \node (M) at (2.5,0) {$\boldsymbol{M} = \begin{pmatrix} M_1 \\ \vdots \\ M_p \end{pmatrix}$};
        \node (Y) at (5,0) {$Y$};

        \draw[->] (M) -- (Y);
        \draw[->] (T) -- (M);
        \draw[->] (T) .. controls (1, -1.7) and (4, -1.7) .. (Y);
    \end{tikzpicture}
}
\newcommand{\MultiMA}{
    \begin{tikzpicture}[node distance=3cm and 3cm, >=stealth]
        \node (T) at (0,0) {$T$};
        \node (Y) at (5,0) {$Y$};
        \node (M1) at (2.5,1.6) {$M_1$};
        \node (M2) at (2.5,1) {$M_2$};
        \node (dots) at (2.5, 0.4)[rotate=90] {...};
        \node (MP2) at (2.5,-0.2) {$M_{p-1}$};
        \node (MP) at (2.5,-0.8) {$M_p$};
        
        \draw[->] (T) -- (M1);
        \draw[->] (T) -- (M2);
        \draw[->] (T) -- (MP2);
        \draw[->] (T) -- (MP);
        \draw[->] (M1) -- (Y);
        \draw[->] (M2) -- (Y);
        \draw[->] (MP2) -- (Y);
        \draw[->] (MP) -- (Y);
        \draw[->] (T) .. controls (1, -1.7) and (4, -1.7) .. (Y);
    \end{tikzpicture}
}

\begin{document}
\title{Joint Adaptive Penalty for Unbalanced Mediation Pathways}
\author[1]{Hanying Jiang\thanks{hjiang252@wisc.edu}}
\author[1]{Kris Sankaran\thanks{ksankaran@wisc.edu}}
\author[1]{Yinqiu He\thanks{yinqiu.he@wisc.edu}}
\affil[1]{Department of Statistics, University of Wisconsin-Madison, USA}
\date{}

\maketitle

\begin{abstract}
Mediation analysis has been widely used to investigate how a treatment influences an outcome through intermediate variables, known as mediators. 
Analyzing a mediation mechanism typically requires assessing multiple model parameters that characterize distinct pathwise effects. 
Classical methods that estimate these parameters individually can be inefficient, particularly when the underlying pathwise effects exhibit substantial imbalance. 
To address this challenge, this work proposes a new joint adaptive penalty that integrates information across entire mediation mechanisms, thereby enhancing both parameter estimation and pathway selection. 
We establish theoretical guarantees for the proposed method under an asymptotic framework and conduct extensive numerical studies to demonstrate its superior performance in scenarios with unbalanced mediation pathways. 
\end{abstract}


\section{Introduction}\label{sec:introduction}

Mediation analysis explores whether and how a treatment influences an outcome through intermediate mediator variables \citep{mackinnon2012introduction,tingley2014mediation}. 
It decomposes the treatment effect on the outcome into the indirect/mediation effect through the mediators and the direct effect through other causal mechanisms. 
This decomposition can improve our understanding of complex mechanisms and inform the design of effective interventions. 
In biomedical research, for example, mediation analysis has been used to examine how exposure to fine particulate matter increases mortality through metabolic and cardiovascular diseases \citep{bai2022chronic}, and how antibiotic treatments influence asthma development through changes in the microbiome \citep{toivonen2021antibiotic}. 
Across various scientific fields, modern technological advances enable the simultaneous measurement of numerous candidate mediators, such as genomic features \citep{abrishamcar2022dna, yang2024causal}, neural signatures \citep{chen2021identifying, zhao2022bayesian}, and metabolic traits \citep{ko2023metabolic, lu2023lipid}. 
Developing efficient analysis for multiple mediators can help uncover active pathways for more targeted interventions and deepen our understanding of complicated systems.

In this work, we focus on the multi-mediator framework where the goal is to understand how an exposure/treatment $T$ influences an outcome $Y$ through $p$ potential mediators $\boldsymbol{M} = [M_1, \ldots, M_p]^\top$. 
To achieve this, classical mediation analyses model the relationship between exposure, potential mediators, and the outcome through directed acyclic graphs. 
One class of existing studies investigates multiple mediators as a group, which can be illustrated as in  Figure \ref{fig:MA_single1}. 
This class includes methods examining the overall group-level mediation effect \citep{vanderweele2014mediation,zhou2020estimation,hao2023simultaneous} or jointly transforming multiple mediators, such as to principal components \citep{huang2016hypothesis,chen2018high,bellavia2019approaches,zhao2020sparse}. 

\begin{figure}
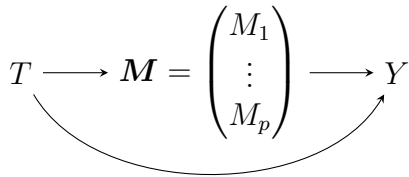
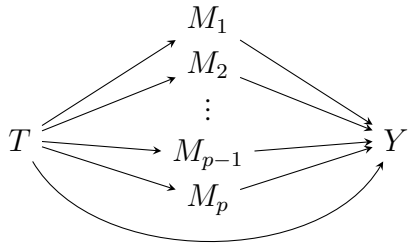

    \centering
    \begin{subfigure}{0.45\textwidth}
        \centering
        \JointMA
        \caption{Group-level mediation pathway.}
        \label{fig:MA_single1}
    \end{subfigure}
    \hfill
    \begin{subfigure}{0.45\textwidth}
        \centering
        \MultiMA
        \caption{Individual mediation pathways.}
        \label{fig:MA_multiple1}
    \end{subfigure}
    \caption{Directed acyclic graphs for mediation pathways.}
    \label{fig:MA1}
\end{figure}

To further reveal the detailed causal mechanisms through the observed mediators, various efforts have also been made to study pathways via each individual mediator \citep{imai2013identification,vansteelandt2017interventional,jerolon2021causal}, as illustrated in Figure \ref{fig:MA_multiple1}. 
In this case, a central question is to estimate the mediation effect through each individual mediator $M_j$ and assess whether it is significantly different from zero, indicating whether $M_j$ plays an active mediating role or not. 
Estimating mediation effects and identifying active pathways through individual mediators can uncover complicated causal mechanisms and guide the development of effective interventions. 
To address the above question, researchers have proposed statistical methods from both frequentist and Bayesian perspectives. 
From the frequentist perspective, pathwise effects are often modeled as fixed parameters, commonly as coefficients in structural equation models. 
Examining mediation effects is then formulated as estimating these coefficients and determining whether or not the corresponding estimands are zero. 
Within this framework, one research line focuses on fitting a joint model of $(T,\boldsymbol{M},Y)$ based on observed data, often incorporating regularizations to handle multiple mediators. 
Examples include the minimax concave penalty \citep{zhang2016estimating}, de-biased LASSO \citep{gao2019testing}, adaptive LASSO \citep{zhang2022high} and pathway-specific penalty \citep{zhao2022pathway}. 
Another research line considers scenarios where estimators for pathwise effects along $T \to M_j$ and $M_j \to Y$ have been obtained along with their asymptotic distributions. 
Then these studies primarily focus on correcting for multiple comparisons across multiple mediators \citep{dai2022multiple,liu2022large,du2023methods}. 
From the Bayesian perspective, pathwise effects are modeled as random coefficients in structural equation models. 
Estimation and variable selection can be achieved through Bayesian shrinkage estimation with appropriate priors  \citep{song2020bayesian,song2021bayesian}.

This work focuses on fitting the {joint} model of $(T, \boldsymbol{M}, Y)$ with regularization from the {frequentist} perspective. 
Within this paradigm, most existing methods apply separate regularizations for the exposure-to-mediator ($T \to \boldsymbol{M}$) and mediator-to-outcome ($\boldsymbol{M} \to Y$) paths \citep{gao2019testing,zhang2022high}. 
However, such a separate-fitting strategy does not effectively combine joint information across the two sets of paths $T \to \boldsymbol{M}$ and $\boldsymbol{M} \to Y$ in mediation analysis. 
This limitation becomes especially problematic if the pathwise effects are \textit{unbalanced}. 
For example, when the pathwise effect along $T\to M_j$ is strong but that along $M_j \to Y$ is weak, examining the two paths separately may result in overlooking this active mediation pathway $T \to M_j\to Y$ as a whole. 
A more detailed illustration under the classical parallel model is provided in Section \ref{sec:challenges}. 
In the existing literature, \cite{zhao2022pathway} introduces a pathway-specific LASSO method and jointly fits the model over $(T, \boldsymbol{M}, Y)$. 
However, it can be computationally expensive, and its accuracy of identifying active pathways can be low, as further demonstrated in the Supporting Information.

This work focuses on fitting the {joint} model of $(T, \boldsymbol{M}, Y)$ with regularization from the {frequentist} perspective. 
Within this paradigm, most existing methods apply separate regularizations for the exposure-to-mediator ($T \to \boldsymbol{M}$) and mediator-to-outcome ($\boldsymbol{M} \to Y$) paths \citep{gao2019testing,zhang2022high}.  
However, such a separate-fitting strategy does not effectively combine joint information across the two sets of paths $T \to \boldsymbol{M}$ and $\boldsymbol{M} \to Y$ in mediation analysis. 
This limitation becomes especially problematic if the pathwise effects are \textit{unbalanced}. 
For example, when the pathwise effect along $T \to M_j$ is strong but that along $M_j \to Y$ is weak, examining the two paths separately may result in overlooking this active mediation pathway $T \to M_j\to Y$ as a whole. 
A more detailed illustration under the classical parallel model is provided in Section \ref{sec:challenges}. 
In the existing literature, \cite{zhao2022pathway} introduces a pathway-specific LASSO method and jointly fits the model over $(T, \boldsymbol{M}, Y)$. 
However, it can be computationally expensive, and its accuracy of identifying active pathways can be low, as further demonstrated in Section \ref{sec:PathwayLASSO} of the appendix. 

To overcome the above challenges, this work proposes a new joint adaptive penalty that combines information across distinct and potentially unbalanced pathwise effects while maintaining low computational cost. 
The penalty is constructed by incorporating adaptive weights informed by the significance of target mediation effects. 
Theoretically, we establish asymptotic guarantees showing that the proposed penalty controls estimation errors and achieves consistent selection of active mediation pathways. 
Through extensive numerical studies, we demonstrate our method is scalable and yields superior performance across various scenarios. 
The new penalization framework will advance current mediation analysis with multiple mediators, facilitating more important scientific discoveries. 

The rest of the paper is organized as follows. 
Section \ref{sec:MA_framework} introduces the framework under which our analysis is conducted. 
Section \ref{sec:method} introduces the new proposed penalty, including its construction and asymptotic theory. 
Section \ref{sec:literature_review} reviews the comparable methods in the existing literature. 
Section \ref{sec:simulations} conducts numerical experiments to compare the proposed method and the existing methods under finite samples. 
In Section \ref{sec:application}, we demonstrate our method by investigating the effect of gastrectomy on total cholesterol level mediated by the gut microbiome. 
We conclude this paper with discussions in Section \ref{sec:discussion}. 
All the proofs are deferred to the appendix. 

We will use the following notations throughout the paper. 
For two sequences of real numbers $\left(a_n\right)$ and $\left(b_n\right)$, 
we let $a_n \gg b_n$ denote $\lim_{n \to \infty} {b_n}/{a_n} = 0$, 
let $a_n \ll b_n$ denote $\lim_{n \to \infty} {a_n}/{b_n} = 0$, 
and let $a_n \lesssim b_n$ denote that there exists a constant $C>0$ such that $|a_n| \leq C |b_n|$ for all $n$. 
For a sequence of random variables $(X_n)$ and a sequence of real numbers $(a_n)$, 
we let $X_n = O_p(a_n)$ represent that for any $\epsilon > 0$, 
there is a positive constant $C_\epsilon$ such that $\sup_n \Pr\left( \left|X_n\right| \geq C_\epsilon |a_n| \right) < \epsilon$. 
We use $\toP$ to denote convergence in probability. 
For a vector $\vect{x}=\left[ x_1,\ldots, x_p \right]^{\top} \in \reals^p$ and a positive integer $q$, 
let $\|\vect{x}\|_q = (\sum_{j = 1}^{p} \left|x_j\right|^q)^{1/q}$ represent the $\ell_q$ norm of $\vect{x}$. 
For a matrix $\vect{A} = \left[a_{ij}\right]\in \reals^{m\times n}$, 
let $\vect{A}^\top$ represent its transpose, 
and let $\Fnorm{\vect{A}} = (\Tr (\vect{A}^\top \vect{A}))^{1/2}$ represent its Frobenius norm. Let  $A \ind B \mid \mathcal{E}$ represent the independence of random variables $A$ and $B$ conditional on an event $\mathcal{E}$.

\section{Model and Setup}\label{sec:MA_framework}

We consider that the exposure $T$, $p$ potential mediators $\boldsymbol{M}$, and outcome $Y$ follow the canonical linear structural equation model \citep{mackinnon2012introduction}: 
\begin{align} 
   \vectrv{M} = \vect{\alpha}^{\ast} T + \vect{\zeta}_M^{\ast \top} \vectrv{X} + \vectrv{E}, \quad \quad 
   \ Y        = \BT^{\ast} T + \vect{\beta}^{\ast \top} \vectrv{M} + \vect{\zeta}_Y^{\ast \top} \vectrv{X} + \epsilon, \label{eq:LSEM}
\end{align}
where 
$\vect{\alpha}^{\ast} = \left[\alpha_1^{\ast}, \ldots, \alpha_p^{\ast}\right]^{\top} \in \mathbb{R}^p$, 
$\vect{\beta}^{\ast} = \left[\beta_1^{\ast}, \ldots, \beta_p^{\ast}\right]^{\top} \in \mathbb{R}^p$, 
$\vect{\zeta}_M^{\ast} = \left[\zeta_{M,ij}^{\ast}\right]_{1 \leq i \leq q, 1 \leq j \leq p} \in \mathbb{R}^{q \times p}$, and 
$\vect{\zeta}_Y^{\ast} = \left[\zeta_{Y,1}^\ast, \ldots, \zeta_{Y, q}^\ast \right]^{\top} \in \mathbb{R}^q$. 
Additionally, $\vectrv{X}$ represents a $q$-dimensional observed pre-treatment confounding variable, and we assume its first element is set to be one to allow for an intercept. 
The random errors $\vectrv{E}$ and $\epsilon$ are independent with zero mean, and $(\vectrv{E}, \epsilon)$ are independent of $(T, \vectrv{X})$. 
We emphasize that the model \eqref{eq:LSEM} is considered for the simplicity of illustration and interpretation, whereas our Joint Adaptive Penalty proposed in Section \ref{sec:method} is general and could potentially be extended under other models for mediation pathway analysis. 

Coefficients in the model \eqref{eq:LSEM} can be connected with causal estimands, particularly individual mediation/indirect effects under the counterfactual framework \citep{imai2013identification,loh2022disentangling}. 
In the existing literature, one class of works examines classical natural indirect effects through individual mediators, and identification typically requires that the causal relationships between multiple mediators be either absent or known \citep{daniel2015causal,taguri2018causal}. 
Another class of studies examines interventional indirect effects, which are defined by setting the mediator to a random draw from the distribution of the counterfactual mediator, and often do not require knowledge of the causal structure among mediators \citep{vansteelandt2017interventional}.
Under the canonical parallel path model \eqref{eq:LSEM}, the analytical forms of the natural and interventional indirect effects coincide under suitable assumptions \citep{jerolon2021causal,loh2022disentangling}. For simplicity, we only review the standard natural indirect effect and present its analytical form under \eqref{eq:LSEM}. 
Notably, the conditions for identifying interventional indirect effects are generally weaker, with further discussions available in  \cite{loh2022disentangling} and \cite{miles2023causal}. 

Let $\vectrv{M}(t) = (M_1(t), \ldots, M_p(t))^{\top}$ represent the potential value of $\vectrv{M}$ under the treatment status $t$. 
For each $j = 1, \ldots, p$, let $\vectrv{M}_{-j}$ denote entries in $\vectrv{M}$ excluding $M_j$, and similarly define $\vectrv{M}_{-j}(t)$. 
Let $Y(t, \vect{m})$ represent the potential outcome of $Y$ if $T$ and $\vectrv{M}$ were set to be $t$ and $\vect{m}$, respectively. 
To define the natural indirect effect through a mediator, we follow the framework in \cite{imai2013identification} which assumes no causal ordering between mediators. 
In this case, we can simultaneously let $(M_j(t'), \vectrv{M}_{-j}(t''))$ denote potential outcomes of $M_j$ and $\vectrv{M}_{-j}$ if $T$ were set to be $t'$ and $t''$, respectively, for $j = 1, \ldots, p$. 
We then consider the natural indirect effect through $M_j$ defined as 
$\delta_j(t'; t) = \E \left[Y(t, M_j(t'), \vectrv{M}_{-j}(t))\right] - \E\left[Y(t, M_j(t), \vectrv{M}_{-j}(t))\right]$, where $t'$ and $t$ are the treatment statuses being compared. 
To identify $\delta_j(t'; t)$, we assume the standard consistency condition (Condition \ref{cond:consistency}) and the sequential ignorability condition for multiple causally unrelated mediators (Condition \ref{cond:simma}) first introduced in \cite{jerolon2021causal}. 

\begin{condition}\label{cond:consistency}  
    For all possible values of $t$ and $\vect{m}$, 
    $\vectrv{M} = \vectrv{M}(t)$ and $Y = Y(t, \vectrv{M}(t))$ if $T = t$, and $Y = Y(t, \vect{m})$ if $T = t$ and $\vectrv{M} = \vect{m}$.
\end{condition}

\begin{condition}\label{cond:simma}   
    For $j = 1, \ldots, p$ and all possible values of $t, t', t'', m$, and $\vect{w}$, 
    assume (i) $\{Y(t, m, \vect{w}), M_j(t'), \vectrv{M}_{-j}(t'') \} \ind T \mid \{\vectrv{X} = \vect{x}\},$ \ 
    (ii) $Y(t', m, \vect{w}) \ind (M_j(t), \vectrv{M}_{-j}(t)) \mid \{T = t, \vectrv{X} = \vect{x}\},$ 
    and (iii) $ Y(t, m, \vect{w}) \ind (M_j(t'), \vectrv{M}_{-j}(t)) \mid  \{T = t, \vectrv{X} = \vect{x}\}$. 
\end{condition} 

Detailed interpretations of Condition \ref{cond:simma} can be found in \cite{jerolon2021causal}. 
Notably, Condition \ref{cond:simma} allows the mediators to be uncausally correlated after conditioning on the treatment and observed pretreatment confounders, e.g., due to unmeasured pretreatment confounders. 
Such flexibility can accommodate mediators that  covary together in real-world applications. 
We next derive an analytical formula of $\delta(t,t')$ in Lemma \ref{lemma:identifiability} below. It generalizes Corollary 3.2 in \cite{jerolon2021causal} by relaxing their assumptions on the Gaussianity and constant correlations of noise terms.

\begin{lemma}\label{lemma:identifiability}
    Under the model \eqref{eq:LSEM} and Conditions \ref{cond:consistency} and \ref{cond:simma}, 
    $\delta_j(t'; t) = \alpha_j^{\ast} \beta_j^{\ast}(t' - t).$  
\end{lemma}

Lemma \ref{lemma:identifiability} shows that $\delta_j(t'; t)$ is proportional to the product of coefficients $\alpha_j^{\ast} \beta_j^{\ast}$ under the model \eqref{eq:LSEM}. 
Therefore, a mediator $M_j$ and its corresponding individual pathway $T \to M_j \to Y$ are referred to as active if $\alpha_j^{\ast} \beta_j^{\ast} \neq 0$. 

\begin{remark}\label{rm:identification}
    In scenarios with multiple mediators, various definitions of indirect effects have been proposed, and the associated identification conditions can differ from and even relax those discussed above, including allowing known and unknown causal orderings between mediators \citep{daniel2015causal,loh2022disentangling}.
    Despite that, the same product-of-coefficients form have been consistently  observed \citep{imai2013identification,daniel2015causal,huang2016hypothesis,loh2022disentangling}. 
    Our proposed method will be based on these products  $\alpha_j^{\ast} \beta_j^{\ast}$ and thus is inherently \textit{generalizable} beyond the particular set of definitions and assumptions reviewed above.  
\end{remark}

\section{Joint Adaptive Penalty}\label{sec:method}

\subsection{Unbalanced Mediation Pathways and Challenges}\label{sec:challenges}

We aim to fit the model \eqref{eq:LSEM} while performing efficient identification on the set of mediators $M_j$ with active mediation effects. 
Under our studied model \eqref{eq:LSEM}, an individual pathway $T \to M_j \to Y$ consists of two paths $T \to M_j$ and $M_j \to Y$ with effects characterized by the two coefficients $\alpha_j^\ast$ and $\beta_j^\ast$, respectively. 
Based on the relative magnitudes between $\alpha_j^{\ast}$ and $\beta_j^{\ast}$, we can divide active pathways into balanced and unbalanced cases. 
Specifically, there can exist three scenarios: 
(a) balanced pathway where the absolute values of $\alpha_j^\ast$ and $\beta_j^\ast$ are similar, 
(b) unbalanced pathway with the scale of $\alpha_j^\ast$ being much larger than that of $\beta_j^\ast$, and 
(c) unbalanced pathway with the scale of $\alpha_j^\ast$ being much smaller than that of $\beta_j^\ast$. 
These scenarios are visualized as in Figure \ref{fig:dag1} (a)--(c). 
As mentioned in Section \ref{sec:introduction}, 
existing methods that apply separate regularizations on the effects $\alpha_j^{\ast}$ and $\beta_j^{\ast}$ may lack power for identifying a significant mediation effect $\alpha_j^{\ast}\beta_j^{\ast}$ under unbalanced scenarios, 
as one of the coefficients can be too small to detect its significance. 

\begin{figure}[!htbp]
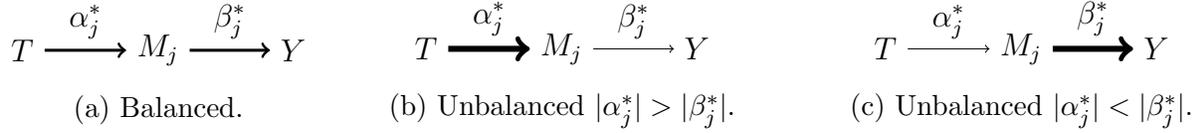

    \centering 
    \begin{subfigure}{0.26\textwidth}
        \centering
        \FigOnesingle
        \caption{Balanced.}
    \end{subfigure}
    \hfill
    \begin{subfigure}{0.35\textwidth}
        \centering
        \FigTwosingle
        \caption{Unbalanced $|\alpha_j^{\ast}|>|\beta_j^{\ast}|$.}
    \end{subfigure}
    \hfill
    \begin{subfigure}{0.35\textwidth}
        \centering
        \FigThreesingle
        \caption{Unbalanced  $|\alpha_j^{\ast}|<|\beta_j^{\ast}|$.}
    \end{subfigure}
    \caption{Visualization of balanced and unbalanced mediation pathways. Line widths of the solid links show relative magnitudes of the corresponding pathwise effects.}
    \label{fig:dag1}
\end{figure} 

\subsection{Procedure} \label{subsec:procedure}

To address the above inherent challenge of identifying effects of unbalanced pathways, we propose a new penalty that can be adaptive to the joint pathway effect of interest, referred to as Joint Adaptive Penalty (\newm{}) below. 
The overall idea is that less regularization shall be used to estimate $\alpha_j^{\ast}$ and $\beta^{\ast}_j$ if a targeted mediation effect $\alpha_j^{\ast} \beta_j^{\ast}$ is significant, and vice versa. 
As true coefficients are unknown in practice, our proposal proceeds in two stages: first, obtain suitable initial estimates of the mediation effects $\alpha_j^{\ast} \beta_j^{\ast}$'s, and second, refit the model with regularization adjusted according to the initial estimates. 
For concreteness, we next describe our proposed method based on the $\ell_1$-norm LASSO penalty of coefficients \citep{tibshirani1996regression,zou2006adaptive}. 
But we emphasize the idea is general and could potentially be extended under other regularizations.  

In particular, consider a dataset with $n$ independently and identically distributed observations $\{T_i, \boldsymbol{M}_i, Y_i, \boldsymbol{X}_i\}_{i = 1}^n$. 
Let $\vT_n = [T_{1}, \ldots, T_{n}]^{\top} \in \reals^n$, $\vM_n = \left[M_{ij} \right]_{1 \leqslant i \leqslant n, 1 \leqslant j \leqslant p} \in \reals^{n \times p}$, $\vY_n = [Y_{1}, \ldots, Y_{n}]^{\top} \in \reals^n$, and $\mathbf{X}_n = \left[X_{ij} \right]_{1 \leqslant i \leqslant n, 1 \leqslant j \leqslant q} \in \reals^{n \times q}$. 
We assume the data following the model \eqref{eq:LSEM}, i.e.,
\begin{equation}\label{eq:model_n}
\begin{aligned}
    \vM_n &= \vT_n \An[\top] + \mathbf{X}_n \boldsymbol{\zeta}_{M}^{*} + \vE_n, \\
    \vY_n &= \vT_n  \BTn + \vM_n \Bn + \mathbf{X}_n \boldsymbol{\zeta}_{Y}^{*} + \vepsilon_n,
\end{aligned}
\end{equation}
where $\vE_n = \left[E_{ij} \right]_{1 \leqslant i \leqslant n, 1 \leqslant j \leqslant p}$ and $\vepsilon_n = [\epsilon_{1}, \ldots, \epsilon_{n}]^\top$ are the error matrix/vector of the exposure-to-mediator and mediator-to-outcome models, whose rows/entries are independently and identically distributed with zero mean. 
For the simplicity of presentation, We let
\begin{align*}
    \mathbf{D}_M = &~ (\vT_n, \mathbf{X}_n) \in \mathbb{R}^{n \times (1 + q)}, \  &\boldsymbol{\theta}_M = &~(\vect{\alpha}, \vect{\zeta}_M^{ \top})^{\top} \in \mathbb{R}^{(1 + q) \times p},\\
    \mathbf{D}_Y = &~ (\vT_n, \mathbf{X}_n, \vM_n) \in \mathbb{R}^{n \times (1 + q + p)}, \  &\boldsymbol{\theta}_Y = &~(\BT, \vect{\zeta}_Y^{ \top}, \vect{\beta}^{ \top})^{\top} \in \mathbb{R}^{1 + q + p}
\end{align*}
represent design matrices and coefficients in the exposure-to-mediator and mediator-to-outcome models, respectively. 
We propose the following procedure that can estimate model coefficients and identify active mediation pathways. 

\medskip
\textbf{Step 1: Initialization.} 
Construct initial estimates 
$\AnhatInit = [\AnjhatInit[1], \ldots, \AnjhatInit[p]]^\top$ and 
$\BnhatInit = [\BnjhatInit[1], \ldots, \BnjhatInit[p]]^\top$ for $\An$ and $\Bn$, and let $\AnjhatInit \BnjhatInit$ be initial estimates for $\alpha_{nj}^{\ast} \beta_{j}^{\ast}$ across $j = 1, \ldots, p$. 

\textbf{Step 2: Joint Adaptive Penalized Regression.} 
Estimate coefficients in the model \eqref{eq:LSEM} by solving 
\begin{align}
    \hat{\boldsymbol{\theta}}_{M, n} = &~ \argmin_{ \boldsymbol{\theta}_M \in \mathbb{R}^{(1 + q) \times p} } \, 
    \ell_{M}(\boldsymbol{\theta}_M; \, \mathbf{D}_M, \mathbf{M}_n) + \lambdaA \sum_{j = 1}^p \frac{|\alpha_{j}|}{\wA}, \label{eq:fit_mediator}\\
    \hat{\boldsymbol{\theta}}_{Y,n} = &~  \argmin_{ \boldsymbol{\theta}_Y \in \mathbb{R}^{1 + q + p} } \, \ell_{Y}( \boldsymbol{\theta}_Y; \, \mathbf{D}_Y, \mathbf{Y}_n) +  \lambdaB \sum_{j = 1}^p \frac{|\beta_{j}|}{\wB}, \label{eq:fit_outcome}
\end{align}
where $\ell_{M}(\cdot)$ and $\ell_{Y}(\cdot)$ represent the loss functions that fit data without regularizations imposed on $\boldsymbol{\alpha}$ and $\boldsymbol{\beta}$, 
$\lambdaA\geqslant 0$ and $ \lambdaB\geqslant 0$ are regularization parameters for $\boldsymbol{\alpha}$ and $\boldsymbol{\beta}$, respectively, 
and we construct pathway adaptive weights 
\begin{align} \label{eq:wweights}
    \wA = |\AnjhatInit \BnjhatInit|^{\gAOne} + |\AnjhatInit|^{2 \gATwo}
    \hspace{1.7em} \text{and} \hspace{1.7em} 
    \wB = |\AnjhatInit \BnjhatInit|^{\gBOne} + |\BnjhatInit|^{2 \gBTwo}, 
\end{align} 
with prespecified constants $\gAOne > 2\gATwo > 0$ and $\gBOne > 2\gBTwo > 0$. 
The proposed joint adaptive penalized estimates $\Anhat = (\hat{\alpha}_{n1},\ldots, \hat{\alpha}_{np})^{\top}$ and $\Bnhat=(\hat{\beta}_{n1},\ldots, \hat{\beta}_{np})^{\top}$ for $\vect{\alpha}^{\ast}$ and $\vect{\beta}^{\ast}$ are constructed as taking the corresponding components in $\hat{\boldsymbol{\theta}}_{M,n} $ and $\hat{\boldsymbol{\theta}}_{Y,n}$.

\textbf{Step 3: Identification of Active Mediation Pathways.}
The proposed penalized estimates can be used to identify active mediation pathways with selection efficiency improved from combining two pathway effects. 
In particular, we construct the set
\begin{align} \label{eq:ashat}
    \AShat_n = \left\{j:\ \Anjhat \Bnjhat \neq 0, j = 1, \ldots, p  \right\}. 
\end{align}

\medskip

The proposed penalties in \eqref{eq:fit_mediator} and \eqref{eq:fit_outcome} extend the classical LASSO penalty and yield tailored estimation for mediation pathway analysis. 
When $\gATwo = \gAOne = \gBTwo = \gBOne = 0$, the weights in \eqref{eq:wweights} reduce to constants, and \eqref{eq:fit_mediator} and \eqref{eq:fit_outcome} become equivalent to fitting LASSO to the exposure-to-mediator and mediator-to-outcome models, respectively. 
Our proposed penalties achieve estimation that is adaptive to mediation pathway properties by incorporating the weights in \eqref{eq:wweights}. 
In particular, each weight in \eqref{eq:wweights} consists of two parts: 
one proportional to the exponential of $|\AnjhatInit \BnjhatInit|$, the absolute value of the initial estimate of the mediation effect, 
and the other proportional to the exponential of the magnitude of the initial estimate of a single coefficient, i.e., $|\AnjhatInit|$ or $|\BnjhatInit|$. 
Without the former, the proposed penalty reduces to the adaptive LASSO \citep{zou2006adaptive}, which adaptively assigns a smaller penalty to a coefficient if its initial estimate is significant. 
Our proposal generalizes the idea by assigning a smaller penalty to a coefficient if either its own initial estimate is significant or the initial estimate of its corresponding mediation effect is significant. 
In this way, if a single coefficient is weak but its corresponding mediation effect is significant, it is less likely to be missed using the proposed penalties compared to using the LASSO or adaptive LASSO. 
This can be especially helpful under the types of unbalanced mediation pathways illustrated in Figure \ref{fig:dag1}. 

The proposed initialization+refitting strategy is advantageous for efficient computation and flexible implementation. 
In particular, the adaptive pathway information $|\AnjhatInit \BnjhatInit|$ from initialization is fixed during the refitting stage, and the refitting optimization in \eqref{eq:fit_mediator} and \eqref{eq:fit_outcome} are convex with respect to $\boldsymbol{\alpha}$ and $\boldsymbol{\beta}$, allowing for the application of a wide range of standard solvers and flexible choice of loss functions. 
On the other hand, although it may be tempting to consider a penalty that is a function of $|\alpha_j \beta_j|$, i.e., directly examining targeted mediation effects without initialization, we note that $|\alpha_j \beta_j|$ is not a convex function with respect to $(\alpha_j, \beta_j)$ as pointed out by \cite{zhao2022pathway}. 
As a result, extra adjustments may be needed and the computation can be burdensome. 

More generally, the proposed adaptive weights may be combined not only with LASSO but also with any other penalties that one prefers. 
As an example, for another penalty $\mathcal{P}(\alpha_j)$ used in fitting the $T \to M_j$ model, we can similarly reweight the penalty as $\mathcal{P}(\alpha_j)/ \wA$ to achieve adaptive properties under unbalanced mediation pathways. 
In the above discussion, we illustrate the proposed idea using $\mathcal{P}(\alpha_j) = |\alpha_j|$ which generalizes LASSO for its popularity and simplicity of implementation. 
We will stick to this choice in the remainder of this paper for ease of presentation, but we expect that similar theoretical and numerical properties can be achieved using other forms of penalties under suitable conditions. 

\begin{remark} \label{rm:extension}
The proposed adaptive weighting strategy can be readily extended beyond the model discussed above, making it a versatile tool under a range of problems with potential unbalanced parameters. 
A key strength of our approach is its ability to borrow statistical efficiency across models through initialization, while preserving low computational cost in the refitting phase. 
This flexibility makes it particularly well-suited for examining target causal effects that depend on multiple model parameters, a common characteristic in mediation pathway models across diverse data types, including compositional data \citep{sohn2019compositional,sohn2022compositional,jiang2024multimedia}, categorical and count data \citep{hao2025class}, and survival outcome \citep{tchetgen2011causal}. 
Furthermore, while the adaptive strategy is demonstrated under two-step pathways as illustrated in Figure \ref{fig:dag1}, it can be similarly generalized to multi-step causal chains involving multiple parameters \citep{shi2022testing}, underscoring its potential as a comprehensive and scalable method for analyzing complex causal effects. 
\end{remark}

\subsection{Implementation}\label{subsec:computation} 

We next discuss the implementation of Steps 1 and 2 in the proposed procedure. 

\paragraph{Step 1:} 
To obtain reliable numerical results, the initial estimates $(\AnjhatInit, \BnjhatInit)$ should adequately approximate the true values $(\alpha_j^{\ast}, \beta_j^{\ast})$, while ensuring the stability of inverse weights in \eqref{eq:fit_mediator} and \eqref{eq:fit_outcome}. 
To this end, we will combine ordinary least squares (OLS) estimates with appropriate lower bounds. 
We find this estimate exhibits both efficient and stable performance through the extensive numerical studies in Section \ref{sec:simulations}. 
In particular, let $\AnjOLS$ and $\BnjOLS$ denote the OLS estimates of $\alpha_j^\ast$ and $\beta_j^\ast$ under the model \eqref{eq:LSEM}. Then we define $\mathcal{T}(x, l) = \sgn(x) (\max\{|x| - l, 0\} + l)$ and construct 
\begin{equation}
\label{eq:truncated_OLS}
    \AnjhatInit = \mathcal{T}(\AnjOLS,\, \TrnjA),
    \quad \quad  \text{and} \quad \quad 
    \BnjhatInit = \mathcal{T}(\BnjOLS, \TrnjB),
\end{equation}
where we set lower bounds $\TrnjA = \CTr \cdot \hat{se}(\AnjOLS)$ and $\TrnjB = \CTr \cdot \hat{se}(\BnjOLS)$ with $\hat{se}(\AnjOLS)$ and $\hat{se}(\BnjOLS)$ denoting the estimated standard error of $\AnjOLS$ and $\BnjOLS$ from the OLS regression. 
This construction ensures that $|\AnjhatInit| \geq \TrnjA$ and $|\BnjhatInit| \geq \TrnjB$, preventing the penalty weights from diverging too fast when OLS estimates $\AnjOLS$ and $\BnjOLS$ approach zero, thereby enhancing robustness to poor OLS estimates. 
A detailed theoretical investigation of the asymptotic properties will be provided in Section \ref{sec:theory}. 
Beyond this specific construction \eqref{eq:truncated_OLS}, other initializations with similar properties could also be used. 

\paragraph{Step 2:} 
In the adaptive penalized regression, loss functions $\ell_{M}(\cdot)$ and $\ell_{Y}(\cdot)$ can be specified by users and take general forms. 
One simple yet effective choice is the quadratic loss function. 
When the dimension of $\boldsymbol{X}$ becomes higher, regularization of coefficients $\boldsymbol{\zeta}_M$ and $\boldsymbol{\zeta}_Y$ may also be added to improve estimation efficiency. 
For instance, we may choose 
\begin{align}\label{eq:quadloss}
    \ell_{M}(\vect{\theta}_M;\, \mathbf{D}_M, \vM_n) = &~\Fnorm{\vM_n - \mathbf{D}_M \vect{\theta}_M}^2 + \mathcal{P}_{M}(\boldsymbol{\zeta}_M), \\
    \ell_{Y}(\vect{\theta}_Y;\, \mathbf{D}_Y, \vY_n) = &~\Fnorm{\vY_n - \mathbf{D}_Y \vect{\theta}_Y}^2 + \mathcal{P}_{Y}(\boldsymbol{\zeta}_Y, \eta), \notag
\end{align} where for $A \in \{M, Y\}$, $\mathcal{P}_{A}(\cdot)$ represents the penalty of the corresponding input coefficients. 
In \eqref{eq:quadloss}, setting $\mathcal{P}_{A}(\cdot) = 0$ gives vanilla quadratic loss functions, and $\mathcal{P}_{A}(\cdot)$ can also be a function with respect to only a subset of input parameters, e.g., the coefficient of the intercept is often not penalized in practice. 
Numerically, our extensive experiments suggest that \eqref{eq:quadloss} could yield sufficiently good and stable empirical results, as will be detailed in Section \ref{sec:simulations}. 
Computationally, under \eqref{eq:quadloss}, optimizations \eqref{eq:fit_mediator} and \eqref{eq:fit_outcome} may be transformed to examining penalized coefficients only, allowing the potential to further reduce computational complexity. 
This is shown by Proposition \ref{proposition:solution} below. 
To facilitate the subsequent presentation, we define $\mathbf{R}_M = \mathbf{M}_n$ and $\mathbf{R}_Y = \mathbf{Y}_n$. 
Moreover, for $A \in \{M,Y\}$, let $\boldsymbol{\theta}_{AP}$ and $\boldsymbol{\theta}_{AU}$ represent penalized and unpenalized coefficients in $\boldsymbol{\theta}_{A}$, respectively, and let $\mathbf{D}_{AP}$ and $\mathbf{D}_{AU}$ be corresponding columns in the design matrix $\mathbf{D}_A$, respectively. 

\begin{proposition}\label{proposition:solution} 
For $A \in \{M,Y\}$, assume $\mathbf{D}_{A}$ has full column rank, and $\mathcal{P}_A(\cdot)$ is convex. 
Solving $(\hat{\boldsymbol{\theta}}_M, \hat{\boldsymbol{\theta}}_Y)$ by \eqref{eq:fit_mediator} and \eqref{eq:fit_outcome} is  equivalent to solving that for $A \in \{M,Y\}$, 
\begin{equation}\label{eq:opt}
\begin{aligned}
    \hat{\boldsymbol{\theta}}_{AP} 
    &= \argmin_{\vect{\theta}_{AP}} \Fnorm{\mathbf{P}_{ {AU}}^{\perp} (\mathbf{R}_A - \mathbf{D}_{AP}\boldsymbol{\theta}_{AP})}^2 + \bar{\mathcal{P}}_A(\boldsymbol{\theta}_{AP}), \\   
    \hat{\boldsymbol{\theta}}_{AU} 
    &= \mathbf{D}_{AU}^{\dagger}(\mathbf{R}_A \mathbf{D}_{AP}\hat{\boldsymbol{\theta}}_{AP}), 
\end{aligned}
\end{equation}
where 
$\mathbf{D}_{AU}^{\dagger} = (\mathbf{D}_{AU}^{\top}\mathbf{D}_{AU})^{-1} \mathbf{D}_{AU}^{\top}$, $\mathbf{P}_{{AU}}^{\perp} = \mathrm{I}_{n\times n} - \mathbf{D}_{AU}\mathbf{D}_{AU}^{\dagger}$, 
$\mathrm{I}_{n\times n}$ denotes an $n\times n$ identity matrix, and 
$\bar{\mathcal{P}}_M(\boldsymbol{\theta}_{MP}) = \lambdaA \sum_{j = 1}^p {|\alpha_{j}|}/{\wA} + \mathcal{P}_M(\boldsymbol{\zeta}_M)$ and 
$\bar{\mathcal{P}}_Y(\boldsymbol{\theta}_{YP}) = \lambdaB \sum_{j = 1}^p {|\beta_{j}|}/{\wB} + \mathcal{P}_Y(\boldsymbol{\zeta}_Y, \eta)$ represent two augmented penalties. 
\end{proposition}

Proposition \ref{proposition:solution} shows that optimizing penalized losses only requires examining $\boldsymbol{\theta}_{AP}$, which typically has a lower dimension than $\boldsymbol{\theta}_A$. 
This reduction in dimensionality might help decrease computational complexity when iterative optimization is needed. 
Moreover, once the penalized coefficient estimate $\hat{\boldsymbol{\theta}}_{AP}$ is obtained, 
unpenalized coefficient estimate $\hat{\boldsymbol{\theta}}_{AU}$ can be computed through closed-form formulae. 
When $\mathcal{P}_A(\cdot)$ takes $\ell_1$-norm based penalty, $\hat{\boldsymbol{\theta}}_{AP}$ can be easily obtained by a standard LASSO-based solver, such as the \textsf{R} package \texttt{glmnet} \citep{friedman2010regularization}. When $\mathcal{P}_M(\cdot)=0$, a closed-form formula for penalized coefficients $\hat{\boldsymbol{\alpha}}$ can in fact be obtained; see Remark \ref{rm:alpha_sol} in Appendix. 
While we discuss the implementation under \eqref{eq:quadloss}, we emphasize that the proposed adaptive strategy in \eqref{eq:fit_mediator} and \eqref{eq:fit_outcome} is general and can be used with any other loss functions.

\section{Asymptotic Theory} \label{sec:theory}

This section establishes asymptotic guarantees for the proposed procedure in Section \ref{sec:method}, giving insights into the advantage of the joint adaptive penalty for unbalanced pathways. We begin by examining properties of the initialization and then analyze the joint adaptive estimators. 

\subsection{Initialization}\label{subsec:initialization}

As discussed in Section \ref{subsec:computation}, ideal initial estimates should accurately approximate the true coefficients with stability near zero. 
This is formalized as Condition \ref{cond:initial} below. 

\begin{condition}[Initialization]\label{cond:initial}
The constructed initial estimates $\AnhatInit$ and $\BnhatInit$ in \eqref{eq:wweights} satisfy
\begin{enumerate}
\setlength{\itemsep}{0pt} 
    \item[(i)] $\sqrt{n}(\AnjhatInit - \Anj) = O_p(1)$ and $\sqrt{n}(\BnjhatInit - \Bnj) = O_p(1)$;
    \item[(ii)] $1/\AnjhatInit = O_p(\sqrt{n})$ and $1/\BnjhatInit = O_p(\sqrt{n})$.
\end{enumerate}
\end{condition}

We will show that the initial estimates in Section \ref{subsec:computation} satisfy Condition \ref{cond:initial}. 
We require the following condition. 

\begin{condition}[Moments]\label{cond:moments}
Assume that  
\begin{itemize}
\setlength{\itemsep}{0pt} 
    \item[(i)] $\mathrm{cov}(\vectrv{E}) = \boldsymbol{\Sigma}$ and $\mathrm{cov}(\epsilon) = \sigma^2$, where $\boldsymbol{\Sigma}$ and $\sigma^2$ are fixed and positive (definite); 
    \item[(ii)] as $n \to \infty$, 
    $\X^\top \X / n \to \SigmaX$, 
    $\|\mathbf{P}_{\mathbf{X}_n}^{\perp} \mathbf{T}_n\|_2^2 / n \to \CA$, and 
    $\max_{1\leq i \leq n} \left(\PXP \vT_n\right)_i^2/n \to 0$, where 
    $\SigmaX$ and $\CA$ are fixed and positive (definite), 
    $\mathbf{P}_{\mathbf{X}_n}^{\perp} = \mathrm{I}_{n \times n} - \X(\X^{\top} \X)^{-1} \X$, and 
    $\left(\PXP \vT_n\right)_i$ represents the $i$th element of $\PXP \vT_n$. 
\end{itemize}
\end{condition}

Condition \ref{cond:moments} imposes regularity conditions on the second-order moments of the error terms and covariates in the model \eqref{eq:model_n}, which are common in the regression analysis \citep{seber2012linear}.  

\begin{proposition}\label{lemma:initial}
    Under the model \eqref{eq:model_n} and Condition \ref{cond:moments}, 
    the initial estimates $\AnhatInit$ and $\BnhatInit$ constructed in \eqref{eq:truncated_OLS} satisfy Condition \ref{cond:initial}.
\end{proposition}

\subsection{Estimation and Selection Consistency}\label{subsec:consistency}

Under the above framework, we define the targeted set of active mediators as the indices in the following set: 
\begin{equation}\label{eq:active_set}
    \AS = \left\{j: \Anj \Bnj \neq 0,  \ \  j = 1,\ldots, p \right\},
\end{equation} 
which can be interpreted as the mediators for which pathways $T \to M_j$ and $M_j \to Y$ have nonzero signals. 
We now study the asymptotics of our method. 
For ease of illustration, we consider that quadratic loss functions are used in \eqref{eq:fit_mediator} and \eqref{eq:fit_outcome}, i.e.,
\begin{equation}\label{eq:quadloss_theory}
\begin{aligned}
    \ell_{M}(\vect{\theta}_M;\, \mathbf{D}_M, \vM_n) 
    & = \Fnorm{\vM_n - \mathbf{D}_M \vect{\theta}_M }^2, \\
    \ell_{Y}(\vect{\theta}_Y;\, \mathbf{D}_Y, \vY_n) 
    &= \Fnorm{\vY_n - \mathbf{D}_Y \vect{\theta}_Y }^2.
\end{aligned}
\end{equation}
These loss functions are common in practice and clarify the essence of signal adaptation achieved by the proposed method. 
More generally, we expect that conclusions for other loss functions can be similarly established given suitable assumptions. 

\begin{theorem}\label{thm:oracle_property}
    Assume Conditions \ref{cond:initial}-\ref{cond:moments} under the model \eqref{eq:model_n}. 
    Suppose the tuning parameters satisfy $n^{1/2 - \gATwo} \ll \lambdaA \ll n^{1/2}$ and $n^{1/2 - \gBTwo} \ll \lambdaB \ll n^{1/2}$,
    where $0 < 2\gATwo < \gAOne$ and $0 < 2\gBTwo < \gBOne$ are specified in  \eqref{eq:wweights}.
    Then with the use of the quadratic loss functions in \eqref{eq:quadloss_theory}, the proposed estimates in \eqref{eq:fit_mediator} and \eqref{eq:fit_outcome} satisfy $\Fnorm{\Anhat-\An}^2 + \Fnorm{\Bnhat - \Bn}^2 \toP 0$. 
    Moreover, $\AShat_n$ constructed in \eqref{eq:ashat} satisfies 
    \begin{equation}\label{eq:selection_consistency}
        \lim_{n \to \infty} \Prob(\AShat_n = \AS) = 1. 
    \end{equation} 
\end{theorem}

Theorem \ref{thm:oracle_property} provides asymptotic guarantees for the joint adaptive penalty in both parameter estimation and active pathway selection, laying the theoretical foundation for its practical utility. 

To further gain insights into how the proposed method adapts to unbalanced pathways,
we compare the proposed joint adaptive penalty weights with those from a standard adaptive LASSO penalty, i.e., comparing weights $\hat{w}_{nj, \alpha}$ and $\hat{w}_{nj, \beta}$ in \eqref{eq:fit_mediator} and \eqref{eq:fit_outcome}  with the weights in adaptive LASSO: $\hat{w}_{\text{AL}, nj, \alpha} = |\AnjhatInit|^{2\eta_{\alpha}}$ and $\hat{w}_{\text{AL}, nj, \beta} = |\BnjhatInit|^{2\eta_{\beta}}$, respectively. 

\begin{proposition}\label{prop:penaltyscale}
Assume Conditions \ref{cond:initial} and \ref{cond:moments}, $0 < 2\gATwo < \gAOne$ and $0 < 2\gBTwo < \gBOne$. Then 
\begin{equation}\label{eq:weight_ratio}
\begin{aligned}
    \frac{\hat{w}_{nj, \alpha}}{\hat{w}_{\mathrm{AL},nj, \alpha}} 
    & \toP 1 + |\alpha_j^\ast|^{\gAOne - 2 \gATwo}|\beta_j^\ast|^{2 \gAOne}, \\ 
    \frac{\hat{w}_{nj,\beta}}{\hat{w}_{\mathrm{AL},nj, \beta}} 
    & \toP 1 + |\alpha_j^\ast|^{2 \gBOne}|\beta_j^\ast|^{\gBOne - 2 \gBTwo}.
\end{aligned}
\end{equation}
\end{proposition}

Proposition \ref{prop:penaltyscale} shows that for an active pathway $T \to M_j \to Y$ where $\Anj$ and $\Bnj$ are nonzero, the two ratios in \eqref{eq:weight_ratio} are greater than $1$. This implies that the proposed joint adaptive penalty is of smaller order compared to the adaptive LASSO penalty. 
In contrast, for a nonactive pathway where at least one of $\Anj$ and $\Bnj$ is zero, both ratios in \eqref{eq:weight_ratio} equal $1$. This implies that the joint adaptive penalty and the adaptive LASSO penalty would remain at the same order. 
Therefore, it is easier to distinguish the active pathways from the remaining ones with the joint adaptive penalty.

\section{Related Methods}  \label{sec:literature_review}

In this section, we review related methods in the literature to set the stage for the numerical comparisons in Section \ref{sec:simulations}. 
Since our proposed method emphasizes regularized estimation of the model \eqref{eq:LSEM} along with the identification of active pathways through individual mediators, we focus on methods with comparable setups and objectives, specifically those addressing the dual goals of estimation and identification. 
It is worth noting that the proposed method is flexible and can incorporate other methodological advances; see Remark \ref{rm:combine}.  
For a broader overview on modern mediation analysis with multiple mediators, we refer readers to comprehensive reviews \citep{clark2023methods,blum2020challenges,zeng2021statistical}.
In the following, we organize our discussion by separately examining frequentist and Bayesian approaches.

Under the frequentist perspective,  mediation effects, i.e., $\alpha_j^*\beta_j^*$  under the model \eqref{eq:LSEM}, are typically treated as fixed parameters. 
One research line utilizes a fitting-and-testing strategy:  first fitting the $T$-$\vectrv M$ and $\vectrv M$-$Y$ models in \eqref{eq:LSEM} separately and obtaining asymptotic distributions of the coefficient estimates; 
second, identifying active mediation pathways by testing hypotheses $H_{0,j}: \alpha_j^{*}\beta_j^*=0$ for $j=1,\ldots, p$ with appropriate adjustments for multiple comparisons.
For example, 
\citet{zhang2016estimating}  first perform an additional initial screening to reduce the number of mediators by applying sure independence screening \citep{fan2008sure} to the $\vectrv M$-$Y$ model. They then fit $T$-$\vectrv M$ and $\vectrv M$-$Y$ models and obtain asymptotic distributions of the coefficient estimates by the ordinary least squares regression and minimax concave penalty regularized regression \citep{zhang2010nearly}, respectively. 
\cite{gao2019testing} propose a similar approach  but obtain $\hat{\beta}_j$ estimates and their asymptotic distributions through debiased LASSO \citep{zhang2014confidence,van2014asymptotically} to mitigate biases. 
\citet{zhang2022high} applies the adaptive LASSO to both $T$-$\vectrv M$ and $\vectrv M$-$Y$ models with asymptotic distributions derived in \cite{zou2006adaptive}. 
As an alternative to the above methods based on separate model fitting, 
\cite{zhao2022pathway} examines the joint model \eqref{eq:LSEM} with a penalty involving $\boldsymbol{\alpha}$ and $\boldsymbol{\beta}$ simultaneously. In particular, it proposes to penalize mediation effects by $\sum_{j=1}^p |\alpha_j\beta_j|$, alongside elastic net regularizations for $\boldsymbol{\alpha}$ and $\boldsymbol{\beta}$ individually \citep{zou2005regularization}. The mediators with $\hat{\alpha}_j\hat{\beta}_j\neq 0$ are identified as active.

\begin{remark}\label{rm:combine}
    While Section \ref{sec:method} focuses on model selection through regularization, our proposed method can be seamlessly combined with other widely used strategies, such as sure independence screening and multiple hypothesis testing, similar to the methods reviewed above. 
    Notably, our theoretical analysis has derived asymptotic  distributions of coefficient estimates, providing the foundation needed to apply established testing procedures for mediation pathways \citep{dai2022multiple,liu2022large,he2024adaptive}. 
    To improve the accuracy of inference under finite sample sizes, future research could explore high-order refinements for uncertainty quantification \citep{chatterjee2013rates}. 
    In current simulations and data analysis, our proposed method in Section \ref{sec:method}  demonstrates sufficiently good performance. 
    For clarity and conciseness, we do not pursue post-selection inference in this paper and leave these intriguing directions for future investigation. 
\end{remark}

Under the Bayesian perspective,  
estimating mediation effects and 
identifying active mediation pathways can be framed as a Bayesian shrinkage estimation problem for the coefficients under \eqref{eq:LSEM}. 
In this vein, \cite{song2020bayesian} specifies separate shrinkage priors for  $\alpha_j$ and $\beta_j$, followed by model selection based on posterior inclusion probabilities. 
Alternatively, \cite{song2021bayesian} proposes joint prior distributions on $(\alpha_j, \beta_j)$, utilizing hard thresholding to specifically target non-zero mediation effects. 

\section{Simulation Studies}\label{sec:simulations}

We next evaluate the finite-sample performance of the proposed method and compare with the methods reviewed in Section \ref{sec:literature_review} through comprehensive simulations. 
As highlighted above, the proposed penalty can be advantageous for identifying active yet unbalanced mediation pathways.  
In the following, Section \ref{subsec:DGM} presents the simulation settings, encompassing both balanced and unbalanced pathways, and provides implementation details for all methods under comparison.
Section \ref{subsec:results} reports numerical results focusing on the accuracy of selecting active mediation pathways,   highlighting the unique advantages of our proposed penalty. 

\subsection{Setup}\label{subsec:DGM}

\paragraph{Data Generation.}
We simulate data under the model \eqref{eq:LSEM}, where $T$ is randomly assigned according to a Bernoulli distribution with success probability $0.5$, 
the mediator-to-outcome-model noise $\epsilon$ is independently drawn from a Gaussian distribution $\Gsn(0, \sigma^2)$, and pret-treatment covariates $\boldsymbol{X}$ are set to be empty for simplicity.  
We consider two cases of exposure-to-mediator-model noises $\boldsymbol{E}$ below to examine the impact of different dependence patterns. 
\begin{itemize}
    \item[(I)] Draw $\boldsymbol{E}$ independently from a multivariate Gaussian distribution $\Gsn\left(\vzero, \SigmaE \right)$ with  $\Sigma_{ij}= (\rho^{|i - j|})$, where $\rho\in \{0, 0.4, 0.8\}$, corresponding to uncorrelated, moderate-correlation, and high-correlation settings. 
    \item[(II)]  After drawing $n$ independent copies of $\boldsymbol{E}$ as in Case (I), randomly permute  $p$ dimensions of the $n$ copies together to allow for correlations between mediators that are non-adjacent in terms of dimension indices. Note $\rho=0$ is excluded in this case as permutation would not change the distribution. 
\end{itemize}
In each case above, we set $p=150$, $\BT^*=1$, and $\sigma^2=1$. We vary the sample size $n \in  \{500, 1000, 1500, 2000\}$ to understand its effect.

For the pairwise coefficients $\{(\alpha_{j}^*,\beta_j^*): j=1,\ldots, p\}$,
we consider six groups of patterns: for each group $k=1,\ldots, 6$, we set 
\begin{equation}\label{eq:coeffgroups}
    (\alpha_{j}^*, \beta_j^*) =\C\cdot ( a_k,  b_k), \ \text{ for } j \in \mathcal{G}_k:= \{ (k - 1)p/6 + 1, \ldots, k p/6\}, 
\end{equation}
where values of $ (a_k, b_k)$  are defined in Table \ref{table:sim_setting},  and $\C^2$ represents the effect magnitude.
\begin{table}[!htbp]
    \setlength{\tabcolsep}{10pt}
    \renewcommand{\arraystretch}{1.5}
    \begin{center}
    \begin{tabular}{c | c c c | c c c} 
     \hline\hline 
      & \multicolumn{3}{c|}{Active} & \multicolumn{3}{c}{Inactive} \\
    Group $k$ & 1 & 2 & 3 & 4 & 5 & 6 \\  
     \hline
     $a_k$ & $1$ & $1/\bSim$ & $\bSim$      & 0 & $1$ & 0\\
     $b_k$ & $1$ & $\bSim$      & $1/\bSim$ & $1$ & 0 & 0\\
     \hline\hline 
    \end{tabular}
    \caption{Patterns of six  groups of coefficient pairs $(\alpha_j^*,\beta_j^*) = C\cdot (a_k, b_k)$ with $|\delta|<1$. }
    \label{table:sim_setting}
    \end{center}
\end{table}

We aim to select the mediators $M_j$ with $j\in \cup_{k=1}^3\mathcal{G}_k$. 
These three active groups satisfy $\alpha_j^*\beta_j^* =C^2 \neq 0$, and correspond to three types of relationships between treatment-mediator and mediator-outcome effects, i.e., 
\begin{align*}
   \frac{\alpha_j^*}{\beta_j^*} =1 \ \text{ for} \ j\in \mathcal{G}_1,     \quad \quad  \frac{\alpha_j^*}{\beta_j^*} = \frac{1}{\delta^2} > 1 \ \text{ for} \ j\in \mathcal{G}_2, \quad \text{ and } \quad   \frac{\alpha_j^*}{\beta_j^*}  = \delta^2<1\ \text{ for} \ j\in \mathcal{G}_3 
\end{align*}
when $|\delta|<1$. 
When fixing the direct effect, the above relationships can give three types of directed acyclic graphs visualized as in Figure \ref{fig:dag3}. 
\begin{figure}[!htbp]
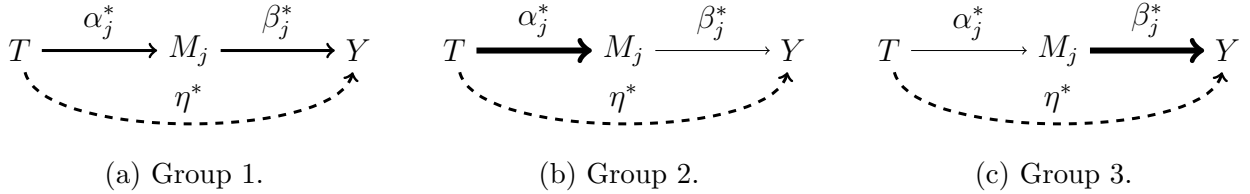

    \centering 
    \begin{subfigure}{0.3\textwidth}
        \centering
        \FigOne
        \caption{Group 1.}
    \end{subfigure}
    \hfill
    \begin{subfigure}{0.3\textwidth}
        \centering
        \FigTwo
        \caption{Group 2.}
    \end{subfigure}
    \hfill
    \begin{subfigure}{0.3\textwidth}
        \centering
        \FigThree
        \caption{Group 3.}
    \end{subfigure}
    \caption{ Directed acyclic graphs of three active groups. Line widths of the solid links indicate relative magnitudes of the corresponding pathwise effects. Dashed line represents the direct effect that is fixed in the analysis.} 
    \label{fig:dag3}
\end{figure} 
As $|\delta|$ becomes smaller, the discrepancy between $\alpha_j^*$ and $\beta_j^*$ becomes larger, which could make it more challenging for methods that examine pathwise effects $\alpha_j^*$ and $\beta_j^*$ separately. 
We aim to exclude the mediators $M_j$ with $j\in \cup_{k=4}^6\mathcal{G}_k$. 
The three inactive groups satisfy $\alpha_j^*\beta_j^* = 0 $ and correspond to three directed acyclic graphs visualized as in Figure \ref{fig:dag2}.
\begin{figure}[!htbp]
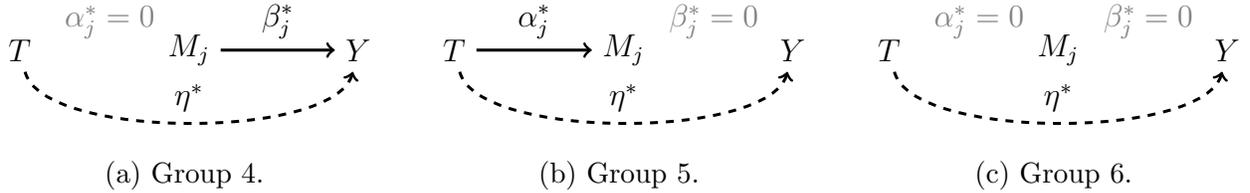

    \centering
    \begin{subfigure}{0.3\textwidth}
        \centering
        \FigFour
        \caption{Group 4.}
    \end{subfigure}
    \hfill
    \begin{subfigure}{0.3\textwidth}
        \centering
        \FigFive
        \caption{Group 5.}
    \end{subfigure}
    \hfill
    \begin{subfigure}{0.3\textwidth}
        \centering
        \FigSix
        \caption{Group 6.}
    \end{subfigure}
    \caption{ Directed acyclic graphs of three inactive groups. The absence of an edge indicates no causal effects.} 
    \label{fig:dag2}
\end{figure}
In our simulations below, we set $C=1$.

\paragraph{Implementation Details.}
We implement the proposed method as discussed in Section \ref{subsec:computation}. For $\CTr$ used in the initialization, our experiments suggest that the results are not sensitive to the choice of $\CTr$,  and we fix $\CTr = 5$ below. Moreover, we use vanilla quadratic loss functions in  \eqref{eq:fit_mediator} and  \eqref{eq:fit_outcome}. 
For the hyperparameters in \eqref{eq:wweights}, 
we consider $\gAOne, \gBOne\in \{0.75, 1, \ldots, 3\}$ and $\gATwo, \gBTwo\in\{0.25, 0.5, 0.75, 1, 1.25\}$, subject to the constraints $\gAOne > 2 \gATwo$ and $\gBOne > 2\gBTwo$. 
For $\lambdaA$ in  \eqref{eq:fit_mediator} and $\lambdaB$ in \eqref{eq:fit_outcome}, we consider exponentially spaced values ranging from $e^0$ to $e^5$ and from $e^3$ to $e^8$, respectively, using a step size of $0.1$ in the exponent. 
We select the hyperparameters $(\gAOne, \gATwo, \lambdaA)$ and $(\gBOne, \gBTwo, \lambdaB)$ by balancing variable selection stability (VSS, \citet{sun2013consistent}) and mean squared error (MSE) in the following sense. 
Given each candidate pair $\boldsymbol{\gamma}'=(\gAOne, \gATwo )$, we first identify the smallest $\lambdaA$ value that achieves the highest VSS computed via $5$-fold cross-validation and denote it as $\lambdaA(\boldsymbol{\gamma}')$. 
We then fit the model for each $(\boldsymbol{\gamma}', \lambdaA(\boldsymbol{\gamma}'))$ and compute MSE on the full data. The final hyperparameters are chosen to minimize MSE. This two-stage procedure aims to achieve good model fit and robust variable selection simultaneously.  
More details on the implementation can be found in \cite{AdaptMLASSOcodes}. 

We compare with the methods reviewed in Section \ref{sec:literature_review}. 
For the ease of reference, we adopt abbreviations of methods as used in \cite{clark2023methods}. 
In particular, HIMA \citep{zhang2016estimating}, HDMA \citep{gao2019testing}, and MedFix \citep{zhang2022high} are conducted using codes 
 by \citet{clark2023methods}. 
BSLMM \citep{song2020bayesian} and PTG \citep{song2021bayesian} are conducted by \texttt{R} package \texttt{bama}. For the frequentist methods involving multiple testing, we choose $0.05$ as the significance level and apply Bonferroni correction. 
For Bayesian methods, we use $0.5$ as the cutoff of posterior inclusion probability. 
All other hyperparameters remain at their default values. 
As Pathway LASSO incurs prohibitive computational costs and demonstrates low selection accuracy in our settings, we present a separate analysis of it in Appendix \ref{sec:PathwayLASSO}. 

In addition, to demonstrate the unique feature of the proposed penalty for mediation pathway selection,
we also compare with two established regularized regression methods for model selection: 
LASSO \citep{tibshirani1996regression} and adaptive LASSO  \cite[][abbreviated as \AL{}]{zou2006adaptive}. 
Under \eqref{eq:fit_mediator} and \eqref{eq:fit_outcome}, 
LASSO corresponds to $\wA =1$ and  $\wB = 1$, and AL   corresponds to  $\wA =  |\AnjhatInit|^{2 \gATwo}$, $\wB = |\BnjhatInit|^{2 \gBTwo}$.
For a fair comparison, we tune the hyperparameters of LASSO and \AL{} using similar strategy to that above for our proposed method and construct selection set same as in \eqref{eq:ashat}. 

\subsection{Numerical Results on Selection Accuracy} \label{subsec:results}

Figures \ref{fig:tuned_unreordered} and  \ref{fig:tuned_reordered} display estimated probabilities of correctly recovering the active set $\AS$ over $100$ Monte Carlo replicates under Cases (I) and (II) of mediator-noise simulations, respectively.
For clear presentation and discussion, we group methods into three classes: penalty-based (LASSO, \AL{}, \newm{}), testing-based (HIMA, HDMA, MedFix), and Bayesian methods (BSLMM, PTG).

The empirical results show that \newm{} achieves the highest accuracy across most regimes. 
In fact, it was only outperformed by HIMA, LASSO, and  \AL{} when $n=500$, $\rho=0.8$, and $\delta$ is close to zero. 
For a fixed $n$, the accuracy of \newm{} increases as $\rho$ becomes smaller. 
For a fixed $\rho$, the accuracy of \newm{} increases as $n$ increases. 
These observations suggest that \newm{} might be advantageous under scenarios with small correlations and large sample sizes.  
Fixing $(n,\rho)$, the selection accuracies of all methods, except PTG, decrease as $\delta$ becomes smaller.
This is reasonable since Table \ref{table:sim_setting} implies that as $\delta$ decreases, $\beta_j^*$ in Group 2 and $\alpha_j^*$ in Group 3 would become smaller, making it more challenging to identify these small nonzero $\alpha_j^*$ and $\beta_j^*$ separately. 
Indeed, all the methods tend to have more false negatives under smaller $\delta$ in our numerical studies.  
Our proposed \newm{} leverages the information on the product effect  $\alpha_j^*\beta_j^*$ and thus can achieve higher accuracy than the other methods do in the challenging small-$\delta$ scenarios. 

For the other two penalty-based methods, LASSO and \AL{}, their relative performances vary across scenarios. 
In Case (I), LASSO tends to outperform \AL{} under a large $n$ and a larger $\rho$, whereas the relationship is at times reversed in Case (II).
The results suggest that LASSO might gain from larger within-group correlations but not between-group correlations. 
For the three testing-based methods,
we observe that HIMA consistently performs the best, while HDMA is close to HIMA under no-correlation and low-correlation regimes but worse than HIMA in high-correlation regimes. MedFix has the lowest selection accuracy among the three in all scenarios. 
Both Bayesian methods, PTG and BSLMM, exhibit relatively low selection accuracies and hardly improve as $n$ or $\bSim$ increases. The underperformance of PTG may be attributed to its design for sparse settings. The accuracy of BSLMM, unlike the other methods, shows non-monotonicity with respect to $\delta$, which may arise from our fixed choice of hyperparameters;  it may perform better if hyperparameters can be chosen in an appropriate data-driven way.

\begin{figure}
    \centering
    \includegraphics[width=\textwidth]{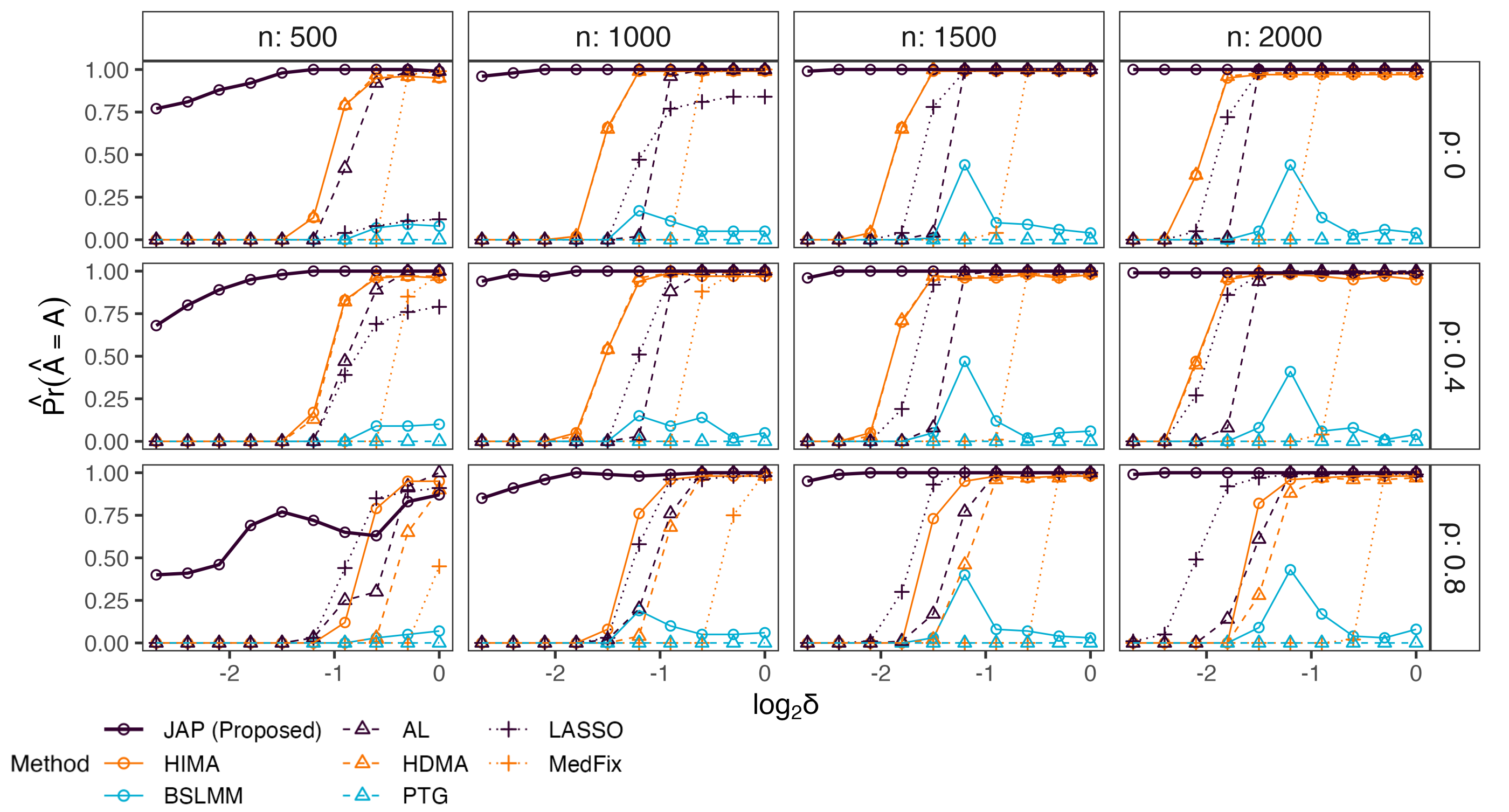}
    \caption{The empirical probability of selecting the active set $\AS$ correctly. The rows correspond to different correlations $\rho$ and the columns correspond to different sample sizes $n$. We use three colors to represent three method classes (penalty-based, testing-based, and Bayesian), while employing distinct linetypes and point shapes to differentiate within each method class. \newm{} is highlighted with a wider line. }
    \label{fig:tuned_unreordered}
\end{figure}

\begin{figure}
    \centering
    \includegraphics[width=\textwidth]{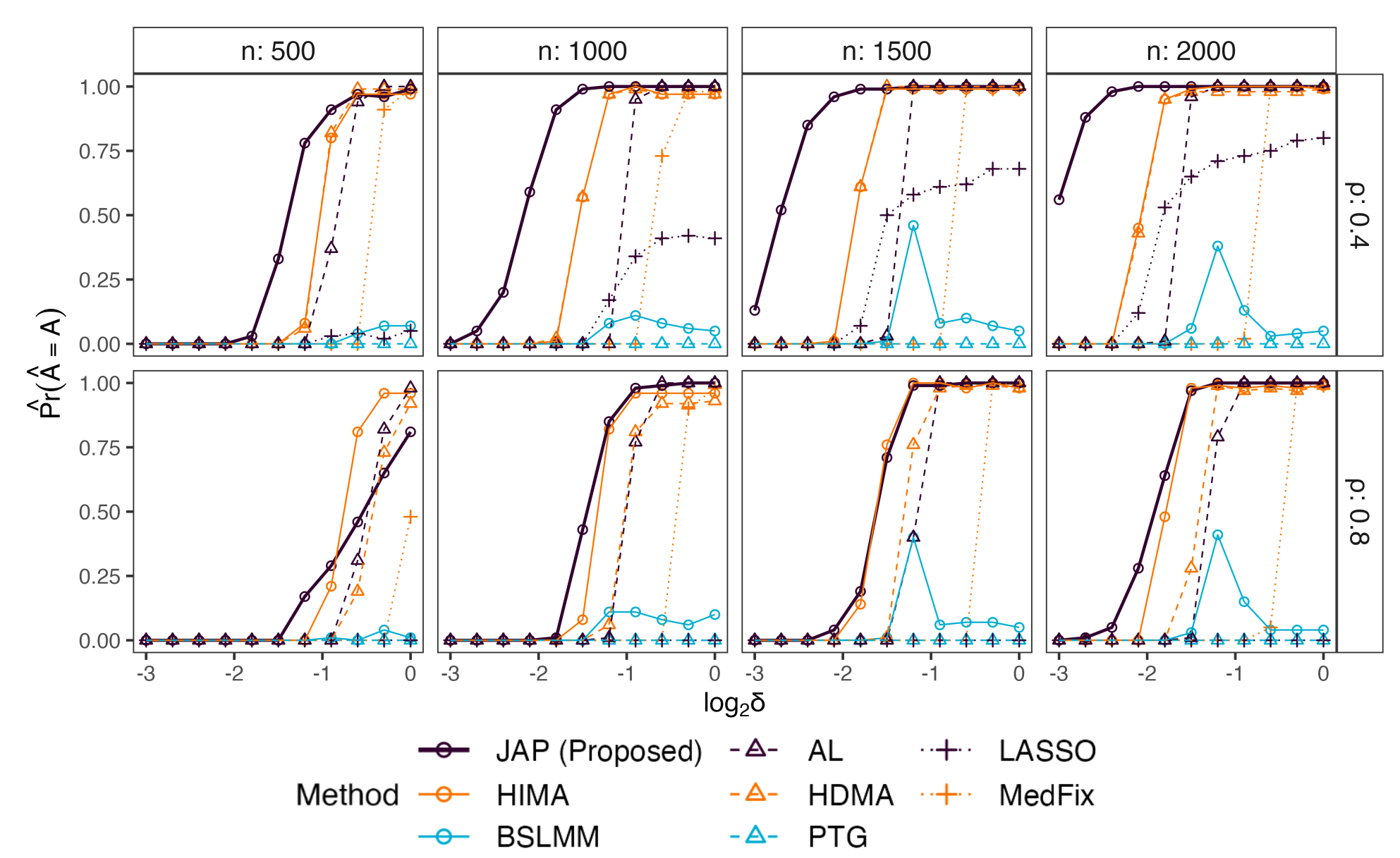}
    \caption{The probability of selecting the active set $\AS$ correctly with randomly reordered mediator model noises $\vE_n$.}
    \label{fig:tuned_reordered}
\end{figure}

\section{Data Analysis}\label{sec:application}

We demonstrate the use of the proposed method by analyzing a gastrectomy dataset from the Curated Gut Microbiome-Metabolome Data Resource \citep{muller2022gut}. 
Our goal is to investigate how the relationships between gastrectomy and the total cholesterol (TC) levels in patients may be mediated through gut microbiome. 
Gastrectomy, the surgical removal of all or part of the stomach, is commonly performed to treat conditions such as gastric cancer \citep{penna2013new}, peptic ulcer \citep{maki1967pylorus}, and morbid obesity \citep{bennett2007surgery}. 
TC is a crucial health indicator and is known to be associated with cardiovascular disease risks, including acute myocardial infarction and stroke \citep{jeong2018effect}. 
Previous studies indicate that TC levels tend to reduce after gastrectomy \citep{tromba2017role,lee2015changes}, and both gastrectomy and TC are related to the gut microbiome \citep{deng2023effect, vourakis2021role, erawijantari2020influence}. 
Mediation analysis can further reveal the role the gut microbiome plays in the process of gastrectomy influencing TC, enhancing our understanding towards the fundamental biological mechanisms underlying the observed effect. 
Specifically, we model gastrectomy as the treatment, TC levels as the outcome, and observed gut microbiome as potential mediators and aim to identify microbiome genera with significant mediation effects. 

In particular, the studied dataset involves patients undergoing total colonoscopy at the National Cancer Center Hospital, Tokyo, Japan. 
Their demographic information, clinical parameters, and faecal samples have been collected. 
More detailed information of data collection can be found in \citet{erawijantari2020influence}. After eliminating the records with missing data, our sample contains $82$ subjects, consisting of $42$ participants who have undergone gastrectomy for gastric cancer and $40$ controls. 
For the gut microbiome, we first filter out genera with a prevalence lower than $0.9$ to concentrate on the core microbiome shared across study subjects.
To accommodate the compositional nature of the count data \citep{gloor2017microbiome}, we apply a centered log-ratio (CLR) transformation with a pseudocount of $1$. 
We then remove the genera with average transformed abundance smaller than $5$, resulting in a total of $25$ genera. 
Throughout the analysis, we adjust for covariates age and gender with $\ell_1$-penalty for their coefficients in the outcome model when fitting the three penalty-based methods. 

\begin{figure*}
    \centering
    \includegraphics[width=1\textwidth]{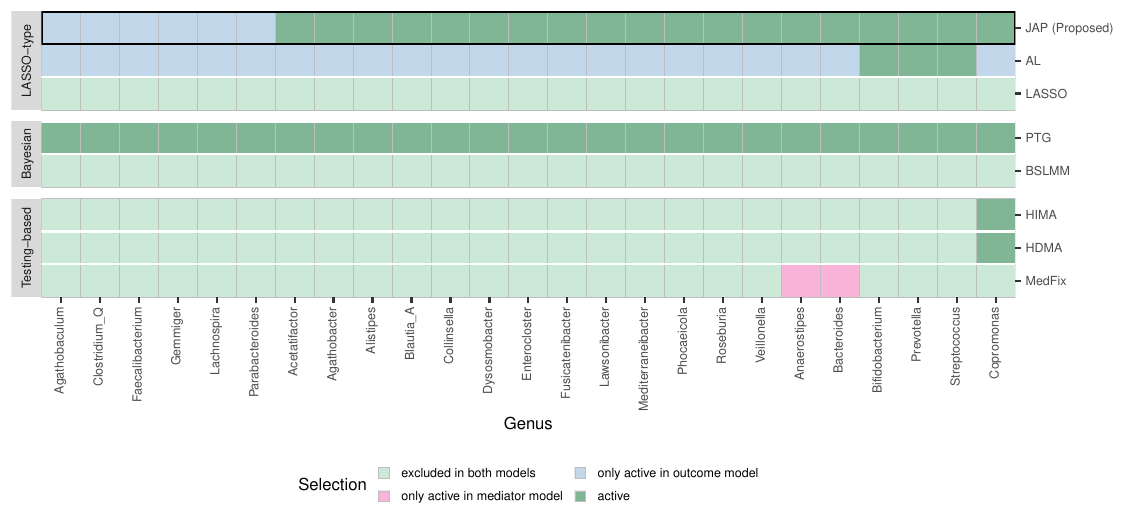}
    \caption{The genera selected by different methods. We use different colors to indicate the four selection results: excluded by both the mediator and outcome models, only included by the mediator model, only included by the outcome model, and active mediators that are included in both models.} 
    \label{fig:gastrectomy}
\end{figure*}

Figure \ref{fig:gastrectomy} displays the results of the eight methods, implemented and organized in three classes as in Section \ref{sec:simulations}. 
For each method (rows) and each microbiome genus $j$ (columns), we use four colors to represent four  types of results: 
(a) both treatment-to-mediator effect $\alpha_j$ and mediator-to-outcome effect $\beta_j$ are significant, indicating the existence of mediation effect $\alpha_j\beta_j$ through the corresponding microbiome genera, (b) only $\alpha_j$ is significant, (c) only $\beta_j$ is significant, and (d) both $\alpha_j$ and $\beta_j$ are not significant. 

In the results, MedFix, BSLMM and LASSO do not identify any significant effects. 
On the other hand, PTG suggests all genera have significant mediation effects, which may be overly optimistic as it tends to give more false positives in our simulations. 
HIMA and HDMA are similar and only identify one genus, \textit{Copromonas}, with significant mediation effect. 
\AL{} suggests all mediator-to-outcome effects $\beta_j$'s are significant, whereas only three genera have significant mediation effects $\alpha_j \beta_j$. 
For those three genera, \textit{Bifidobacterium} and \textit{Prevotella} have been discovered to be associated with the gastrectomy and TC \citep{lin2018long,sanchez2019gut,jia2023gut,deng2023effect,vourakis2021role,yu2020fecal}, and the fluctuation of \textit{Streptococcus} abundance after gastrectomy has also been observed before \citep{yu2020fecal}. 

The proposed method \newm{} identifies $19$ genera with significant mediation effects, including \textit{Roseburia} and \textit{Bacteroides}, whose relationships with gastrectomy and TC have been reported in the literature \citep{yu2020fecal, kissmann2024increase, deng2023effect, wu2022gut, jia2023gut}. 
Compared to the other methods, it can yield more discoveries while maintaining some discernment. 
Also, \newm{} is similar to \AL{} in terms of identifying all mediator-to-outcome effects $\beta_j$'s, which is understandable by the connections between their constructions. 
But \newm{} identifies more treatment-to-mediator effects $\alpha_j$'s than \AL{} does, which could be attributed to the effective use of joint pathwise information in \newm{}. 
For example, existing studies have shown that the abundances of \textit{Bacteroides} and \textit{Veillonella} changed after gastrectomy \citep{yu2020fecal, kissmann2024increase, deng2023effect, wu2022gut}. 
Their treatment-to-mediator effects $\alpha_j$'s have been identified by \newm{} but not \AL{}. 
The results suggest that the proposed \newm{} can be a powerful method in practice. 

\section{Discussion}
\label{sec:discussion}

In this work, we propose a novel joint adaptive penalty for regularized mediation analysis.
Our approach incorporates adaptive weights informed by the significance of mediation effects to improve statistical efficiency in estimating and identifying unbalanced mediation pathways. 
Theoretically, we establish rigorous asymptotic guarantee for controlling the estimation error and consistent selection of active mediation pathways. 
Numerically, we demonstrate the adaptability and scalability of the proposed method across diverse scenarios.  
Our proposed strategy provides a flexible and powerful framework for analyzing mediation effects, opening several new avenues for future  research.

First,  as noted in Remark \ref{rm:extension}, the adaptive weights can be extended to mediation analysis models for  diverse data types and causal chains \citep{jiang2024multimedia,hao2025class,tchetgen2011causal}. 
While our core idea of incorporating the significance of target causal effects remains applicable across different models and targets, 
the constructions and performance of extensions would require case-by-case investigation, presenting important questions and opportunities for future exploration. 

Second, although this paper considers scenarios with fixed-dimensional mediators for the ease of
understanding and illustration, we anticipate that our theoretical results can be extended to accommodate high-dimensional mediators under appropriate assumptions \citep{huang2008adaptive}. 
Methodologically, the proposed method could potentially be enhanced by combining with other widely used strategies, such as sure independence screening \citep{fan2008sure}. 
Developing appropriate refinements and comprehensive understanding for the proposed adaptive weights in high-dimensional scenarios would be important areas for future research. 

Third, this paper focuses on model estimation and consistent selection, where achieving the latter typically requires strong assumptions on signal strengths. More generally, alternative measures of variable selection performance, such as false discovery rate, could also be considered. 
This may be achieved by developing post-selection inference tools for the proposed adaptive penalty and combining them with multiple testing methods \citep{liu2022large}, as discussed in Remark \ref{rm:combine}. 
Understanding how the improved estimation efficiency influences the final performance of pathway identification would be an interesting future research direction.  

\section*{Funding}

This research was supported by the NSF Grant DMS-2515523 and Wisconsin Alumni Research Foundation to YH, and the National Institute of General Medical Sciences of the National Institutes of Health under award number R01GM152744 to KS. 

\section*{Acknowledgement}

Computational resources were provided by the Center for High Throughput Computing 
at the University of Wisconsin-Madison \citep{https://doi.org/10.21231/gnt1-hw21}.

\bibliographystyle{apalike}
\bibliography{references}

\begin{thebibliography}{}

\bibitem[Abrishamcar et~al., 2022]{abrishamcar2022dna}
Abrishamcar, S., Chen, J., Feil, D., Kilanowski, A., Koen, N., Vanker, A., Wedderburn, C.~J., Donald, K.~A., Zar, H.~J., Stein, D.~J., et~al. (2022).
\newblock Dna methylation as a potential mediator of the association between prenatal tobacco and alcohol exposure and child neurodevelopment in a south african birth cohort.
\newblock {\em Translational psychiatry}, 12(1):418.

\bibitem[Bai et~al., 2022]{bai2022chronic}
Bai, L., Benmarhnia, T., Chen, C., Kwong, J.~C., Burnett, R.~T., van Donkelaar, A., Martin, R.~V., Kim, J., Kaufman, J.~S., and Chen, H. (2022).
\newblock Chronic exposure to fine particulate matter increases mortality through pathways of metabolic and cardiovascular disease: insights from a large mediation analysis.
\newblock {\em Journal of the American Heart Association}, 11(22):e026660.

\bibitem[Bellavia et~al., 2019]{bellavia2019approaches}
Bellavia, A., James-Todd, T., and Williams, P.~L. (2019).
\newblock Approaches for incorporating environmental mixtures as mediators in mediation analysis.
\newblock {\em Environment international}, 123:368--374.

\bibitem[Bennett et~al., 2007]{bennett2007surgery}
Bennett, J.~M., Mehta, S., and Rhodes, M. (2007).
\newblock Surgery for morbid obesity.
\newblock {\em Postgraduate medical journal}, 83(975):8--15.

\bibitem[Blum et~al., 2020]{blum2020challenges}
Blum, M.~G., Valeri, L., Fran{\c{c}}ois, O., Cadiou, S., Siroux, V., Lepeule, J., and Slama, R. (2020).
\newblock Challenges raised by mediation analysis in a high-dimension setting.
\newblock {\em Environmental health perspectives}, 128(5):055001.

\bibitem[{Center for High Throughput Computing}, 2006]{https://doi.org/10.21231/gnt1-hw21}
{Center for High Throughput Computing} (2006).
\newblock Center for high throughput computing.

\bibitem[Chatterjee and Lahiri, 2013]{chatterjee2013rates}
Chatterjee, A. and Lahiri, S.~N. (2013).
\newblock Rates of convergence of the adaptive lasso estimators to the oracle distribution and higher order refinements by the bootstrap.
\newblock {\em The Annals of Statistics}, 41(3):1232 -- 1259.

\bibitem[Ch{\'e}n et~al., 2021]{chen2021identifying}
Ch{\'e}n, O.~Y., Cao, H., Phan, H., Nagels, G., Reinen, J.~M., Gou, J., Qian, T., Di, J., Prince, J., Cannon, T.~D., et~al. (2021).
\newblock Identifying neural signatures mediating behavioral symptoms and psychosis onset: High-dimensional whole brain functional mediation analysis.
\newblock {\em NeuroImage}, 226:117508.

\bibitem[Ch{\'e}n et~al., 2018]{chen2018high}
Ch{\'e}n, O.~Y., Crainiceanu, C., Ogburn, E.~L., Caffo, B.~S., Wager, T.~D., and Lindquist, M.~A. (2018).
\newblock High-dimensional multivariate mediation with application to neuroimaging data.
\newblock {\em Biostatistics}, 19(2):121--136.

\bibitem[Clark-Boucher et~al., 2023]{clark2023methods}
Clark-Boucher, D., Zhou, X., Du, J., Liu, Y., Needham, B.~L., Smith, J.~A., and Mukherjee, B. (2023).
\newblock Methods for mediation analysis with high-dimensional dna methylation data: Possible choices and comparisons.
\newblock {\em Plos Genetics}, 19(11):e1011022.

\bibitem[Dai et~al., 2022]{dai2022multiple}
Dai, J.~Y., Stanford, J.~L., and LeBlanc, M. (2022).
\newblock A multiple-testing procedure for high-dimensional mediation hypotheses.
\newblock {\em Journal of the American Statistical Association}, 117(537):198--213.

\bibitem[Daniel et~al., 2015]{daniel2015causal}
Daniel, R.~M., De~Stavola, B.~L., Cousens, S.~N., and Vansteelandt, S. (2015).
\newblock Causal mediation analysis with multiple mediators.
\newblock {\em Biometrics}, 71(1):1--14.

\bibitem[Deng et~al., 2023]{deng2023effect}
Deng, C., Pan, J., Zhu, H., and Chen, Z.-Y. (2023).
\newblock Effect of gut microbiota on blood cholesterol: A review on mechanisms.
\newblock {\em Foods}, 12(23):4308.

\bibitem[Du et~al., 2023]{du2023methods}
Du, J., Zhou, X., Clark-Boucher, D., Hao, W., Liu, Y., Smith, J.~A., and Mukherjee, B. (2023).
\newblock Methods for large-scale single mediator hypothesis testing: Possible choices and comparisons.
\newblock {\em Genetic epidemiology}, 47(2):167--184.

\bibitem[Erawijantari et~al., 2020]{erawijantari2020influence}
Erawijantari, P.~P., Mizutani, S., Shiroma, H., Shiba, S., Nakajima, T., Sakamoto, T., Saito, Y., Fukuda, S., Yachida, S., and Yamada, T. (2020).
\newblock Influence of gastrectomy for gastric cancer treatment on faecal microbiome and metabolome profiles.
\newblock {\em Gut}, 69(8):1404--1415.

\bibitem[Fan and Lv, 2008]{fan2008sure}
Fan, J. and Lv, J. (2008).
\newblock Sure independence screening for ultrahigh dimensional feature space.
\newblock {\em Journal of the Royal Statistical Society Series B: Statistical Methodology}, 70(5):849--911.

\bibitem[Friedman et~al., 2010]{friedman2010regularization}
Friedman, J., Hastie, T., and Tibshirani, R. (2010).
\newblock Regularization paths for generalized linear models via coordinate descent.
\newblock {\em Journal of statistical software}, 33(1):1.

\bibitem[Gao et~al., 2019]{gao2019testing}
Gao, Y., Yang, H., Fang, R., Zhang, Y., Goode, E.~L., and Cui, Y. (2019).
\newblock Testing mediation effects in high-dimensional epigenetic studies.
\newblock {\em Frontiers in genetics}, 10:1195.

\bibitem[Geyer, 1994]{geyer1994asymptotics}
Geyer, C.~J. (1994).
\newblock On the asymptotics of constrained {M}-estimation.
\newblock {\em The Annals of statistics}, pages 1993--2010.

\bibitem[Gloor et~al., 2017]{gloor2017microbiome}
Gloor, G.~B., Macklaim, J.~M., Pawlowsky-Glahn, V., and Egozcue, J.~J. (2017).
\newblock Microbiome datasets are compositional: and this is not optional.
\newblock {\em Frontiers in microbiology}, 8:2224.

\bibitem[Hao et~al., 2025]{hao2025class}
Hao, W., Chen, C., and Song, P. X.-K. (2025).
\newblock A class of directed acyclic graphs with mixed data types in mediation analysis.
\newblock {\em Canadian Journal of Statistics}, page e70016.

\bibitem[Hao and Song, 2023]{hao2023simultaneous}
Hao, W. and Song, P.~X. (2023).
\newblock A simultaneous likelihood test for joint mediation effects of multiple mediators.
\newblock {\em Statistica Sinica}, 33(4):2305--2326.

\bibitem[He et~al., 2024]{he2024adaptive}
He, Y., Song, P.~X., and Xu, G. (2024).
\newblock Adaptive bootstrap tests for composite null hypotheses in the mediation pathway analysis.
\newblock {\em Journal of the Royal Statistical Society Series B: Statistical Methodology}, 86(2):411--434.

\bibitem[Huang et~al., 2008]{huang2008adaptive}
Huang, J., Ma, S., and Zhang, C.-H. (2008).
\newblock Adaptive lasso for sparse high-dimensional regression models.
\newblock {\em Statistica Sinica}, pages 1603--1618.

\bibitem[Huang and Pan, 2016]{huang2016hypothesis}
Huang, Y.-T. and Pan, W.-C. (2016).
\newblock Hypothesis test of mediation effect in causal mediation model with high-dimensional continuous mediators.
\newblock {\em Biometrics}, 72(2):402--413.

\bibitem[Imai and Yamamoto, 2013]{imai2013identification}
Imai, K. and Yamamoto, T. (2013).
\newblock Identification and sensitivity analysis for multiple causal mechanisms: Revisiting evidence from framing experiments.
\newblock {\em Political Analysis}, 21(2):141--171.

\bibitem[Jeong et~al., 2018]{jeong2018effect}
Jeong, S.-M., Choi, S., Kim, K., Kim, S.~M., Lee, G., Park, S.~Y., Kim, Y.-Y., Son, J.~S., Yun, J.-M., and Park, S.~M. (2018).
\newblock Effect of change in total cholesterol levels on cardiovascular disease among young adults.
\newblock {\em Journal of the American Heart Association}, 7(12):e008819.

\bibitem[J{\'e}rolon et~al., 2021]{jerolon2021causal}
J{\'e}rolon, A., Baglietto, L., Birmel{\'e}, E., Alarcon, F., and Perduca, V. (2021).
\newblock Causal mediation analysis in presence of multiple mediators uncausally related.
\newblock {\em The International Journal of Biostatistics}, 17(2):191--221.

\bibitem[Jia et~al., 2023]{jia2023gut}
Jia, B., Zou, Y., Han, X., Bae, J.-W., and Jeon, C.~O. (2023).
\newblock Gut microbiome-mediated mechanisms for reducing cholesterol levels: implications for ameliorating cardiovascular disease.
\newblock {\em Trends in Microbiology}, 31(1):76--91.

\bibitem[Jiang, 2025]{AdaptMLASSOcodes}
Jiang, H. (2025).
\newblock \url{https://github.com/hhhanying/AdaptMLASSO}.

\bibitem[Jiang et~al., 2024]{jiang2024multimedia}
Jiang, H., Miao, X., Thairu, M.~W., Beebe, M., Grupe, D.~W., Davidson, R.~J., Handelsman, J., and Sankaran, K. (2024).
\newblock Multimedia: multimodal mediation analysis of microbiome data.
\newblock {\em Microbiology Spectrum}, 0(0):e01131--24.

\bibitem[Kissmann et~al., 2024]{kissmann2024increase}
Kissmann, A.-K., Pa{\ss}, F., Ruzicka, H.-M., Dorst, I., Stieger, K.~R., Weil, T., Gihring, A., Elad, L., Knippschild, U., and Rosenau, F. (2024).
\newblock An increase in prominent probiotics represents the major change in the gut microbiota in morbidly obese female patients upon bariatric surgery.
\newblock {\em Women}, 4(1):86--104.

\bibitem[Ko et~al., 2023]{ko2023metabolic}
Ko, J., Sequeira, I.~R., Skudder-Hill, L., Cho, J., Poppitt, S.~D., and Petrov, M.~S. (2023).
\newblock Metabolic traits affecting the relationship between liver fat and intrapancreatic fat: a mediation analysis.
\newblock {\em Diabetologia}, 66(1):190--200.

\bibitem[Lee et~al., 2015]{lee2015changes}
Lee, J.~W., Kim, E.~Y., Yoo, H.~M., Park, C.~H., and Song, K.~Y. (2015).
\newblock Changes of lipid profiles after radical gastrectomy in patients with gastric cancer.
\newblock {\em Lipids in Health and Disease}, 14:1--9.

\bibitem[Lin et~al., 2018]{lin2018long}
Lin, X.-H., Huang, K.-H., Chuang, W.-H., Luo, J.-C., Lin, C.-C., Ting, P.-H., Young, S.-H., Fang, W.-L., Hou, M.-C., and Lee, F.-Y. (2018).
\newblock The long term effect of metabolic profile and microbiota status in early gastric cancer patients after subtotal gastrectomy.
\newblock {\em PLoS One}, 13(11):e0206930.

\bibitem[Liu et~al., 2022]{liu2022large}
Liu, Z., Shen, J., Barfield, R., Schwartz, J., Baccarelli, A.~A., and Lin, X. (2022).
\newblock Large-scale hypothesis testing for causal mediation effects with applications in genome-wide epigenetic studies.
\newblock {\em Journal of the American Statistical Association}, 117(537):67--81.

\bibitem[Loh et~al., 2022]{loh2022disentangling}
Loh, W.~W., Moerkerke, B., Loeys, T., and Vansteelandt, S. (2022).
\newblock Disentangling indirect effects through multiple mediators without assuming any causal structure among the mediators.
\newblock {\em Psychological Methods}, 27(6):982.

\bibitem[Lu et~al., 2023]{lu2023lipid}
Lu, S., Xie, Q., Kuang, M., Hu, C., Li, X., Yang, H., Sheng, G., Xie, G., and Zou, Y. (2023).
\newblock Lipid metabolism, bmi and the risk of nonalcoholic fatty liver disease in the general population: evidence from a mediation analysis.
\newblock {\em Journal of Translational Medicine}, 21(1):192.

\bibitem[MacKinnon, 2012]{mackinnon2012introduction}
MacKinnon, D. (2012).
\newblock {\em Introduction to statistical mediation analysis}.
\newblock Routledge.

\bibitem[Maki et~al., 1967]{maki1967pylorus}
Maki, T., Shiratori, T., Hatafuku, T., and Sugawara, K. (1967).
\newblock Pylorus-preserving gastrectomy as an improved operation for gastric ulcer.
\newblock {\em Surgery}, 61(6):838--845.

\bibitem[Miles, 2023]{miles2023causal}
Miles, C.~H. (2023).
\newblock On the causal interpretation of randomised interventional indirect effects.
\newblock {\em Journal of the Royal Statistical Society Series B: Statistical Methodology}, 85(4):1154--1172.

\bibitem[Muller et~al., 2022]{muller2022gut}
Muller, E., Algavi, Y.~M., and Borenstein, E. (2022).
\newblock The gut microbiome-metabolome dataset collection: a curated resource for integrative meta-analysis.
\newblock {\em npj Biofilms and Microbiomes}, 8(1):79.

\bibitem[Penna and Allum, 2013]{penna2013new}
Penna, M. and Allum, W. (2013).
\newblock New treatments for gastric cancer: are they changing clinical practice?
\newblock {\em Clinical Practice}, 10(5):649.

\bibitem[S{\'a}nchez-Alcoholado et~al., 2019]{sanchez2019gut}
S{\'a}nchez-Alcoholado, L., Guti{\'e}rrez-Repiso, C., G{\'o}mez-P{\'e}rez, A.~M., Garc{\'\i}a-Fuentes, E., Tinahones, F.~J., and Moreno-Indias, I. (2019).
\newblock Gut microbiota adaptation after weight loss by roux-en-y gastric bypass or sleeve gastrectomy bariatric surgeries.
\newblock {\em Surgery for Obesity and Related Diseases}, 15(11):1888--1895.

\bibitem[Seber and Lee, 2003]{seber2012linear}
Seber, G.~A. and Lee, A.~J. (2003).
\newblock {\em Linear regression analysis}.
\newblock John Wiley \& Sons.

\bibitem[Shi and Li, 2022]{shi2022testing}
Shi, C. and Li, L. (2022).
\newblock Testing mediation effects using logic of boolean matrices.
\newblock {\em Journal of the American Statistical Association}, 117(540):2014--2027.

\bibitem[Sohn and Li, 2019]{sohn2019compositional}
Sohn, M.~B. and Li, H. (2019).
\newblock Compositional mediation analysis for microbiome studies.
\newblock {\em The Annals of Applied Statistics}, 13(1):661--681.

\bibitem[Sohn et~al., 2022]{sohn2022compositional}
Sohn, M.~B., Lu, J., and Li, H. (2022).
\newblock A compositional mediation model for a binary outcome: application to microbiome studies.
\newblock {\em Bioinformatics}, 38(1):16--21.

\bibitem[Song et~al., 2021]{song2021bayesian}
Song, Y., Zhou, X., Kang, J., Aung, M.~T., Zhang, M., Zhao, W., Needham, B.~L., Kardia, S.~L., Liu, Y., Meeker, J.~D., et~al. (2021).
\newblock Bayesian sparse mediation analysis with targeted penalization of natural indirect effects.
\newblock {\em Journal of the Royal Statistical Society Series C: Applied Statistics}, 70(5):1391--1412.

\bibitem[Song et~al., 2020]{song2020bayesian}
Song, Y., Zhou, X., Zhang, M., Zhao, W., Liu, Y., Kardia, S.~L., Roux, A. V.~D., Needham, B.~L., Smith, J.~A., and Mukherjee, B. (2020).
\newblock Bayesian shrinkage estimation of high dimensional causal mediation effects in omics studies.
\newblock {\em Biometrics}, 76(3):700--710.

\bibitem[Sun et~al., 2013]{sun2013consistent}
Sun, W., Wang, J., and Fang, Y. (2013).
\newblock Consistent selection of tuning parameters via variable selection stability.
\newblock {\em The Journal of Machine Learning Research}, 14(1):3419--3440.

\bibitem[Taguri et~al., 2018]{taguri2018causal}
Taguri, M., Featherstone, J., and Cheng, J. (2018).
\newblock Causal mediation analysis with multiple causally non-ordered mediators.
\newblock {\em Statistical methods in medical research}, 27(1):3--19.

\bibitem[Tchetgen~Tchetgen, 2011]{tchetgen2011causal}
Tchetgen~Tchetgen, E.~J. (2011).
\newblock On causal mediation analysis with a survival outcome.
\newblock {\em The international journal of biostatistics}, 7(1):0000102202155746791351.

\bibitem[Tibshirani, 1996]{tibshirani1996regression}
Tibshirani, R. (1996).
\newblock Regression shrinkage and selection via the lasso.
\newblock {\em Journal of the Royal Statistical Society Series B: Statistical Methodology}, 58(1):267--288.

\bibitem[Tingley et~al., 2014]{tingley2014mediation}
Tingley, D., Yamamoto, T., Hirose, K., Keele, L., and Imai, K. (2014).
\newblock Mediation: R package for causal mediation analysis.
\newblock {\em Journal of Statistical Software}.

\bibitem[Toivonen et~al., 2021]{toivonen2021antibiotic}
Toivonen, L., Schuez-Havupalo, L., Karppinen, S., Waris, M., Hoffman, K.~L., Camargo~Jr, C.~A., Hasegawa, K., and Peltola, V. (2021).
\newblock Antibiotic treatments during infancy, changes in nasal microbiota, and asthma development: population-based cohort study.
\newblock {\em Clinical Infectious Diseases}, 72(9):1546--1554.

\bibitem[Tromba et~al., 2017]{tromba2017role}
Tromba, L., Tartaglia, F., Carbotta, S., Sforza, N., Pelle, F., Colagiovanni, V., Carbotta, G., Cavaiola, S., and Casella, G. (2017).
\newblock The role of sleeve gastrectomy in reducing cardiovascular risk.
\newblock {\em Obesity surgery}, 27:1145--1151.

\bibitem[van~de Geer et~al., 2014]{van2014asymptotically}
van~de Geer, S., B{\"u}hlmann, P., Ritov, Y., and Dezeure, R. (2014).
\newblock On asymptotically optimal confidence regions and tests for high-dimensional models.
\newblock {\em The Annals of Statistics}, pages 1166--1202.

\bibitem[VanderWeele and Vansteelandt, 2014]{vanderweele2014mediation}
VanderWeele, T. and Vansteelandt, S. (2014).
\newblock Mediation analysis with multiple mediators.
\newblock {\em Epidemiologic methods}, 2(1):95--115.

\bibitem[Vansteelandt and Daniel, 2017]{vansteelandt2017interventional}
Vansteelandt, S. and Daniel, R.~M. (2017).
\newblock Interventional effects for mediation analysis with multiple mediators.
\newblock {\em Epidemiology}, 28(2):258--265.

\bibitem[Vourakis et~al., 2021]{vourakis2021role}
Vourakis, M., Mayer, G., and Rousseau, G. (2021).
\newblock The role of gut microbiota on cholesterol metabolism in atherosclerosis.
\newblock {\em International journal of molecular sciences}, 22(15):8074.

\bibitem[Wu et~al., 2022]{wu2022gut}
Wu, C., Zhao, Y., Zhang, Y., Yang, Y., Su, W., Yang, Y., Sun, L., Zhang, F., Yu, J., Wang, Y., et~al. (2022).
\newblock Gut microbiota specifically mediates the anti-hypercholesterolemic effect of berberine (bbr) and facilitates to predict bbr’s cholesterol-decreasing efficacy in patients.
\newblock {\em Journal of advanced research}, 37:197--208.

\bibitem[Yang et~al., 2024]{yang2024causal}
Yang, H., Liu, Z., Wang, R., Lai, E.-Y., Schwartz, J., Baccarelli, A.~A., Huang, Y.-T., and Lin, X. (2024).
\newblock Causal mediation analysis for integrating exposure, genomic, and phenotype data.
\newblock {\em Annual Review of Statistics and Its Application}, 12.

\bibitem[Yu et~al., 2020]{yu2020fecal}
Yu, D., Shu, X.-O., Howard, E.~F., Long, J., English, W.~J., and Flynn, C.~R. (2020).
\newblock Fecal metagenomics and metabolomics reveal gut microbial changes after bariatric surgery.
\newblock {\em Surgery for Obesity and Related Diseases}, 16(11):1772--1782.

\bibitem[Zeng et~al., 2021]{zeng2021statistical}
Zeng, P., Shao, Z., and Zhou, X. (2021).
\newblock Statistical methods for mediation analysis in the era of high-throughput genomics: current successes and future challenges.
\newblock {\em Computational and structural biotechnology journal}, 19:3209--3224.

\bibitem[Zhang, 2010]{zhang2010nearly}
Zhang, C.-H. (2010).
\newblock Nearly unbiased variable selection under minimax concave penalty.
\newblock {\em The Annals of Statistics}, pages 894--942.

\bibitem[Zhang and Zhang, 2014]{zhang2014confidence}
Zhang, C.-H. and Zhang, S.~S. (2014).
\newblock Confidence intervals for low dimensional parameters in high dimensional linear models.
\newblock {\em Journal of the Royal Statistical Society Series B: Statistical Methodology}, 76(1):217--242.

\bibitem[Zhang et~al., 2016]{zhang2016estimating}
Zhang, H., Zheng, Y., Zhang, Z., Gao, T., Joyce, B., Yoon, G., Zhang, W., Schwartz, J., Just, A., Colicino, E., et~al. (2016).
\newblock Estimating and testing high-dimensional mediation effects in epigenetic studies.
\newblock {\em Bioinformatics}, 32(20):3150--3154.

\bibitem[Zhang, 2022]{zhang2022high}
Zhang, Q. (2022).
\newblock High-dimensional mediation analysis with applications to causal gene identification.
\newblock {\em Statistics in biosciences}, 14(3):432--451.

\bibitem[Zhao et~al., 2022]{zhao2022bayesian}
Zhao, Y., Chen, T., Cai, J., Lichenstein, S., Potenza, M.~N., and Yip, S.~W. (2022).
\newblock Bayesian network mediation analysis with application to the brain functional connectome.
\newblock {\em Statistics in medicine}, 41(20):3991--4005.

\bibitem[Zhao et~al., 2020]{zhao2020sparse}
Zhao, Y., Lindquist, M.~A., and Caffo, B.~S. (2020).
\newblock Sparse principal component based high-dimensional mediation analysis.
\newblock {\em Computational statistics \& data analysis}, 142:106835.

\bibitem[Zhao and Luo, 2022]{zhao2022pathway}
Zhao, Y. and Luo, X. (2022).
\newblock Pathway lasso: pathway estimation and selection with high-dimensional mediators.
\newblock {\em Statistics and its interface}, 15(1):39--50.

\bibitem[Zhou et~al., 2020]{zhou2020estimation}
Zhou, R.~R., Wang, L., and Zhao, S.~D. (2020).
\newblock Estimation and inference for the indirect effect in high-dimensional linear mediation models.
\newblock {\em Biometrika}, 107(3):573--589.

\bibitem[Zou, 2006]{zou2006adaptive}
Zou, H. (2006).
\newblock The adaptive lasso and its oracle properties.
\newblock {\em Journal of the American statistical association}, 101(476):1418--1429.

\bibitem[Zou and Hastie, 2005]{zou2005regularization}
Zou, H. and Hastie, T. (2005).
\newblock Regularization and variable selection via the elastic net.
\newblock {\em Journal of the Royal Statistical Society Series B: Statistical Methodology}, 67(2):301--320.

\end{thebibliography}

\appendix

\begin{center}
   \huge \textbf{Appendix} \normalsize 
\end{center}

We provide the proofs of Lemma \ref{lemma:identifiability}, Propositions \ref{proposition:solution}--\ref{prop:penaltyscale}, 
and Theorem \ref{thm:oracle_property} in Sections \ref{sec:proof_identifiability}--\ref{sec:pf_penaltyscale}, respectively. The supplementary simulation results of Pathway LASSO is presented in Appendix \ref{sec:PathwayLASSO}.

\section{Proof of Lemma \ref{lemma:identifiability}}
\label{sec:proof_identifiability}

To prove Lemma \ref{lemma:identifiability}, we first state conclusions that will be needed in the following derivation.  
In particular, under the model \eqref{eq:LSEM}, we have
\begin{align}
    \E[M_j \mid T = t, \vectrv{X}= \vect{x}]    &= \alpha_j^{\ast} t  + \vect{\zeta}_{Mj}^{\ast\top}\vect{x}, \quad  j = 1, \ldots, p, \label{eq:expt_M}\\
    \E[Y \mid T = t, \vectrv{M} = \vect{m}, \vectrv{X}= \vect{x}] 
    &= \BT^\ast t + \vect{\beta}^{\ast\top} \vect{m} + \vect{\zeta}_Y^{\ast \top} \vect{x}, \label{eq:expt_Y}
\end{align}
where $\vect{\zeta}_{Mj}^{\ast}$ represents the $j$th column of $\vect{\zeta}_{M}^{\ast}$. 
For $ j = 1, \ldots, p$, we have
\begin{align}
        \E[M_j(t) \mid \vectrv{X}= \vect{x}] 
    & = \E[M_j(t) \mid T = t, \vectrv{X}= \vect{x}] 
        \hspace{2em} \text{(by Condition \ref{cond:simma} (i))} \label{eq:expt_Mt}\\
    & = \E[M_j \mid T = t, \vectrv{X}= \vect{x}]  
        \hspace{2em} \text{(by Condition \ref{cond:consistency})} \nonumber\\
    & = \alpha_j^{\ast} t + \vect{\zeta}_{Mj}^{\ast\top}\vect{x}. 
        \hspace{2em} \text{(by \eqref{eq:expt_M})} \nonumber
\end{align}

To derive $\delta_j(t';,t)$,  note that $\E\{Y\left(t, M_j(t'), \vectrv{M}_{-j}(t)\right)] = \E[\E[Y\left(t, M_j(t'), \vectrv{M}_{-j}(t)\right) \mid \boldsymbol{X}]\}$.
We first examine  
\begin{align}  
   &~ \E[Y\left(t, M_j(t'), \vectrv{M}_{-j}(t)\right) \mid \boldsymbol{X}=\boldsymbol{x}]\notag\\
=&~\int_{\reals^p} \E[Y(t, m, \vect{w}) \mid M_j(t') = m, \vectrv{X}= \vect{x}, \vectrv{M}_{-j}(t) = \vect{w}] \diffF_{t',t\mid \vect{x}}(m, \vect{w}), \label{eq:expt_Ytm_1}
\end{align} 
where $F_{t',t\mid \vect{x}}$ represents the distribution  of $(M_j(t'), \vectrv{M}_{-j}(t))$ conditional on $\vectrv{X} = \vect{x}$. 
Since Condition \ref{cond:simma} (i) implies $ Y(t, m, \vect{w}) \ind T \mid \left\{M_j(t') = m', \vectrv{M}_{-j}(t'')= \vect{w}'', \vectrv{X}= \vect{x} \right\}$, 
\begin{align}
    \eqref{eq:expt_Ytm_1} =&~   \int_{\reals^p} \E[Y(t, m, \vect{w}) \mid  M_j(t') = m,   \vectrv{M}_{-j}(t) = \vect{w}, T=t , \, \vectrv{X}= \vect{x}] \diffF_{t',t\mid \vect{x}}(m, \vect{w}) \notag\\
    =&~   \int_{\reals^p} \E[Y(t, m, \vect{w}) \mid  M_j(t) = m,    \vectrv{M}_{-j}(t) = \vect{w}, T=t, \, \vectrv{X}= \vect{x}] \diffF_{t',t\mid \vect{x}}(m, \vect{w}),\label{eq:expt_Ytm_2} 
\end{align}
where the last equation follows by  Condition \ref{cond:simma} (iii). 
By Condition \ref{cond:consistency}, 
\begin{align}
  \eqref{eq:expt_Ytm_2} =  &~ \int_{\reals^p} \E[Y \mid M_j = m, \vectrv{M}_{-j} = \vect{w}, T = t, \vectrv{X}= \vect{x}] \diffF_{t',t\mid \vect{x}}(m, \vect{w}) \\
      =&   ~  \int_{\reals^p} \left(\BT^{\ast} t + \beta_j^{\ast} m +  \vbeta_{-j}^{\ast\top}\vect{w} + \vect{\zeta}_Y^{\ast \top} \vect{x}\right) \diffF_{t',t\mid \vect{x}}(m, \vect{w})  
       \hspace{2em} \text{(by \eqref{eq:expt_Y})} \nonumber\\
   =  &  ~  \BT^{\ast} t + \alpha_j^{\ast} \beta_j^{\ast} t' + \vect{\beta}^{\ast \top}_{-j}\vect{\alpha}^{\ast}_{-j}  t + (\vect{\beta}^{\ast \top} \vect{\zeta}_M^{\ast \top} + \vect{\zeta}_Y^{\ast \top}) \vect{x},  \hspace{2em} \text{(by \eqref{eq:expt_Mt})} \label{eq:expt_Ytm_3}  
\end{align}
where $\vect{\alpha}^{\ast}_{-j}$ and $\vect{\beta}^{\ast}_{-j}$ represent the vectors $\vect{\alpha}^{\ast}$  and $\vect{\beta}^{\ast}$ excluding their $j$-th elements, respectively. 
Combining \eqref{eq:expt_Ytm_1}--\eqref{eq:expt_Ytm_3}, we obtain   $ \E \left[Y(t, M_j(t'), \vectrv{M}_{-j}(t))\right] = \BT^{\ast} t + \alpha_j^{\ast} \beta_j^{\ast} t' + \vect{\beta}^{\ast \top}_{-j}\vect{\alpha}^{\ast}_{-j}  t + (\vect{\beta}^{\ast \top} \vect{\zeta}_M^{\ast \top} + \vect{\zeta}_Y^{\ast \top})\E(\vect{X})$. Similar conclusion can be derived for $\E\left[Y(t, M_j(t), \vectrv{M}_{-j}(t)) \right]$ by letting $t'=t$. In summary, 
 $     \delta_j(t'; t) 
    =\E \left[Y(t, M_j(t'), \vectrv{M}_{-j}(t))\right]  - \E\left[Y(t, M_j(t), \vectrv{M}_{-j}(t)) \right] =   \alpha_j^{\ast} \beta_j^{\ast} (t'-t)$ is obtained. 

\section{Proof of Proposition \ref{proposition:solution}}\label{sec:proof_proposition}  

When applying the loss functions defined in \eqref{eq:quadloss}, we can rewrite the optimizations \eqref{eq:fit_mediator} and \eqref{eq:fit_outcome} as
\begin{equation}
\label{eq:opt_general}
    \left( \hat{\boldsymbol{\theta}}_{AP},  \hat{\boldsymbol{\theta}}_{AU} \right) = \argmin_{\boldsymbol{\theta}_{AP}, \boldsymbol{\theta}_{AU}} \Fnorm{
    \mathbf{R}_A - \mathbf{D}_{AP}\boldsymbol{\theta}_{AP} - \mathbf{D}_{AU}\boldsymbol{\theta}_{AU}
    }^2 + \bar{\mathcal{P}}_A(\boldsymbol{\theta}_{AP})
\end{equation}
for $A\in \{M,Y\}$. 
Since $\mathbf{P}_{{AU}}^{\perp}$ is a projection matrix onto the column space orthogonal to that of $\mathbf{D}_{AU}$, by  properties of projection matrices \citep{seber2012linear}, we can decompose 
$\Fnorm{\mathbf{R}_A - \mathbf{D}_{AP}\boldsymbol{\theta}_{AP} - \mathbf{D}_{AU}\boldsymbol{\theta}_{AU}}^2 = \ell_{A,1}(\boldsymbol{\theta}_{AP})+ \ell_{A,2}( \boldsymbol{\theta}_{AP}, \boldsymbol{\theta}_{AU})$, where $\mathbf{P}_{{AU}}=\mathrm{I}_{n\times n}-\mathbf{P}_{{AU}}^{\perp}$, 
\begin{align*}
    \ell_{A,1}(\boldsymbol{\theta}_{AP})= &~  \Fnorm{\mathbf{P}_{ {AU}}^{\perp} (\mathbf{R}_A-\mathbf{D}_{AP}\boldsymbol{\theta}_{AP}) }^2, \\
    \ell_{A,2}( \boldsymbol{\theta}_{AP}, \boldsymbol{\theta}_{AU}) = &~  \Fnorm{\mathbf{P}_{{AU}}( \mathbf{R}_A - \mathbf{D}_{AP}\boldsymbol{\theta}_{AP}) - \mathbf{D}_{AU}\boldsymbol{\theta}_{AU}}^2, 
\end{align*}
and $ \ell_{A,1}(\cdot)$ does not depend on $\boldsymbol{\theta}_{AU}$. 
Plugging the above decomposition into \eqref{eq:opt_general}, we obtain
\begin{align}
    \left( \hat{\boldsymbol{\theta}}_{AP},  \hat{\boldsymbol{\theta}}_{AU} \right)  
    = &~\argmin_{\boldsymbol{\theta}_{AP}, \boldsymbol{\theta}_{AU}} \ell_{A,1}(\boldsymbol{\theta}_{AP}) + \ell_{A,2}( \boldsymbol{\theta}_{AP}, \boldsymbol{\theta}_{AU})  + \bar{\mathcal{P}}_A(\boldsymbol{\theta}_{AP}), 
\end{align}
which is unique by the convexity of the penalized loss function. 

By equivalently optimizing over $\hat{\boldsymbol{\theta}}_{AU}$ and $\hat{\boldsymbol{\theta}}_{AP}$ sequentially, we have
\begin{align}
     \hat{\boldsymbol{\theta}}_{AP} =&~\argmin_{\boldsymbol{\theta}_{AP}} \,   \ell_{A,1}(\boldsymbol{\theta}_{AP})  + \bar{\mathcal{P}}_A(\boldsymbol{\theta}_{AP}) + \min_{\boldsymbol{\theta}_{AU} }\ell_{A,2}( \boldsymbol{\theta}_{AP}, \boldsymbol{\theta}_{AU}) \notag\\
     =&~ \argmin_{\boldsymbol{\theta}_{AP}} \,   \ell_{A,1}(\boldsymbol{\theta}_{AP})  + \bar{\mathcal{P}}_A(\boldsymbol{\theta}_{AP}),  \label{eq:thetahatp_pf}
\end{align}
where the second equation follows by the property of ordinary least squares regression that 
$\min_{\boldsymbol{\theta}_{AU} }\ell_{A,2}( \boldsymbol{\theta}_{AP}, \boldsymbol{\theta}_{AU})=0$ for each given $\boldsymbol{\theta}_{AP}$. 
Therefore, the first equation in  \eqref{eq:opt} is proved. 
Similarly, by sequential optimization, we know 
\begin{align*}
    \hat{\boldsymbol{\theta}}_{AU} =&~\argmin_{\boldsymbol{\theta}_{AU}} \,   \ell_{A,1}(\hat{\boldsymbol{\theta}}_{AP})  + \bar{\mathcal{P}}_A(\hat{\boldsymbol{\theta}}_{AP}) +  \ell_{A,2}( \hat{\boldsymbol{\theta}}_{AP}, \boldsymbol{\theta}_{AU}) \\
    =&~\argmin_{\boldsymbol{\theta}_{AU}} \ell_{A,2}( \hat{\boldsymbol{\theta}}_{AP}, \boldsymbol{\theta}_{AU}) = \mathbf{D}_{AU}^{\dagger}(\mathbf{R}_A-\mathbf{D}_{AP}\hat{\boldsymbol{\theta}}_{AP}),
\end{align*} 
which follows by the solution of ordinary least squares regression. 

\begin{remark} \label{rm:alpha_sol}
When  $\mathcal{P}_M(\cdot)=0$, $\Anhat$ is given by
\begin{equation}
\label{eq:anjhat_form}
    \Anjhat = \softthre\left(
    \frac{\vM_{nj}^\mytrans \mathbf{P}_{\vect{X}_n}^{\perp}\vT_n} {\|\mathbf{P}_{\vect{X}_n}^{\perp}\vT_n\|_2^2}; 
    \ \frac{\lambdaA}{2 \wA \|\mathbf{P}_{\vect{X}_n}^{\perp}\vT_n\|_2^2}\right), j =  1,\ldots, p,
\end{equation}
where 
$\mathbf{M}_{nj}$ is the $j$th column of $\vM_n$ 
and $\mathbf{P}_{\vect{X}_n}^{\perp} = \mathrm{I}_{n\times n}- \vect{X}_n(\vect{X}_n^\top \vect{X}_n)^{-1}\vect{X}_n^{\top}$.
\end{remark}
\begin{proof}
    By Proposition \ref{proposition:solution}, optimization \eqref{eq:fit_mediator} can be written as
    \begin{align*}
        \Anhat 
        &= \argmin_{\vect{\alpha}\in \reals^p} 
        \Fnorm{\mathbf{P}_{\vect{X}_n}^{\perp} (\vM_n - \vT_n \vect{\alpha}^\top)}^2 +
        \lambdaA \sum_{j = 1}^p \frac{|\alpha_j|}{\wA}\\
        &= \argmin_{\vect{\alpha}\in \reals^p} \sum_{j}^p \left\{\|\mathbf{P}_{\vect{X}_n}^{\perp}\vM_{nj} - \alpha_j \mathbf{P}_{\vect{X}_n}^{\perp}\vT_n\|^2_2 + \lambdaA \frac{|\alpha_j|}{\wA} \right\}.
    \end{align*}
    Hence, each element $\Anjhat$ in $\Anhat$ is the solution to the following optimization problem:
    \begin{equation*}
        \Anjhat = \argmin_{\alpha_j\in \reals} \{\|\mathbf{P}_{\vect{X}_n}^{\perp}\vM_{nj} - \alpha_j \mathbf{P}_{\vect{X}_n}^{\perp}\vT_n\|^2_2 + \lambdaA {|\alpha_j|}/{\wA}\},
    \end{equation*}
    which has a closed form given by \eqref{eq:anjhat_form}.
\end{proof}

\section{Proof of Proposition \ref{lemma:initial}}\label{sec:proof_initial}

First, to show $\AnjhatInit$ constructed in \eqref{eq:truncated_OLS} satisfies Condition \ref{cond:initial}(ii), note that $|\AnjhatInit| \leq \TrnjA = \CTr \cdot \hat{se}(\AnjOLS)$. Thus, $1/\AnjhatInit = O_p(1/\hat{se}(\AnjOLS)) = O_p(\sqrt{n})$, which follows from the third term in \eqref{eq:seAn} in Lemma \ref{lemma:OLS_se} below. 

\begin{lemma} \label{lemma:OLS_se}
    Under the conditions of Proposition \ref{lemma:initial}, OLS estimators $\AnjOLS$ and $\BnjOLS$ satisfy 
    \begin{align}
          \sqrt{n}(\AnjOLS - \Anj) &= O_p(1), &\sqrt{n}\hat{se}(\AnjOLS) &= O_p(1), & \frac{1}{\sqrt{n}\hat{se}(\AnjOLS) }&= O_p(1), \label{eq:seAn}\\
        \sqrt{n}(\BnjOLS - \Bnj)& = O_p(1), &  \sqrt{n}\hat{se}(\BnjOLS) &= O_p(1), & \frac{1}{\sqrt{n} \hat{se}(\BnjOLS)} &= O_p(1). \label{eq:seBn}
    \end{align}
\end{lemma}
\begin{proof}
    See Section \ref{subsec:proof_lemma_OLS_se}.
\end{proof}

Second, we show that $\AnjhatInit$ satisfies Condition \ref{cond:initial}(i). 
For $\AnjhatInit$ in \eqref{eq:truncated_OLS}, 
we can equivalently rewrite it as $ \AnjhatInit = \AnjOLS \indifun_{\{|\AnjOLS|\geq \TrnjA\}} + \TrnjA \indifun_{\{|\AnjOLS|< \TrnjA\}},$ where $\indifun$ represents the indicator function. 
Therefore,
\begin{align}
 & ~\AnjhatInit - \Anj \notag\\
    =&  ~ (\AnjOLS - \Anj)\indifun_{\{|\AnjOLS|\geq \TrnjA\}} + \{\TrnjA - \AnjOLS + (\AnjOLS - \Anj)\} \indifun_{\{|\AnjOLS|< \TrnjA\}}\notag\\
    =&~ \AnjOLS - \Anj + (\TrnjA - \AnjOLS ) \indifun_{\{|\AnjOLS|< \TrnjA\}}. \label{eq:n_consist1}
\end{align} 
Thus, 
\begin{equation}
\label{eq:n_consist2}
    |\sqrt{n}(\AnjhatInit - \Anj)| 
    \leq  \sqrt{n}|\AnjOLS - \Anj|  + 2 \sqrt{n}\TrnjA 
    = \sqrt{n}|\AnjOLS - \Anj|  + 2 \CTr\sqrt{n} \hat{se}(\AnjOLS).
\end{equation}
By the first and second terms in \eqref{eq:seAn} in Lemma \ref{lemma:OLS_se}, we can conclude that $\eqref{eq:n_consist2} = O_p(1)$. 
We can similarly prove that $\BnjhatInit$ satisfies Condition \ref{cond:initial}.

\subsection{Proof of Lemma \ref{lemma:OLS_se}}
\label{subsec:proof_lemma_OLS_se}

\paragraph{Proof of \eqref{eq:seAn}.}
We first show $\sqrt{n}(\AnjOLS - \Anj) = O_p(1)$. 
We use $\vE_{nj}$ to denote the $j$th column of $\vE_{n}$. 
Note that $ \AnjOLS $ is the OLS for the coefficient of $\vT_n$ under the linear regression  $\vM_{nj}\sim (\vT_n, \mathbf{X}_n)$. 
By Frisch–Waugh–Lovell Theorem, we know
\begin{align} 
     \AnjOLS 
    = 
    &~  (\vT_n^{\top} \PXP \vT_n)^{-1} \vT_n^{\top} \PXP \vM_{nj}
    =\Anj + (\vT_n^{\top} \PXP \vT_n)^{-1} \vT_n^{\top} \PXP \vE_{nj}, \label{eq:AnjOLS}
\end{align}
where the second equation follows by the model  $\mathbf{M}_{nj} = \vT_n \Anj + \mathbf{X}_n 
\boldsymbol{\zeta}_{M,j}+ \vE_{nj}$ in  \eqref{eq:model_n}. 
By the independence between $\vE_n$ and $\DM$, and Condition \ref{cond:moments}, we have that
\begin{equation} \label{eq:conv_E_PXP_T_rootn}
    \frac{\vE_n^{\top}\PXP \vT_n}{\sqrt{n}} \tod \Gsn\left(\vzero_p, \sigma^2_T\SigmaE
    \right),
\end{equation}
where $\tod$ denotes convergence in distribution.
Then we have
\begin{equation*}
    \sqrt{n}(\AnjOLS - \Anj) = \left(\frac{\|\mathbf{P}_{\mathbf{X}_n}^{\perp} \mathbf{T}_n\|_2^2}{n}\right)^{-1}  \frac{\vE_{nj}^{\top}\PXP \vT_n}{\sqrt{n}} = O_p(1).
\end{equation*}

We next prove the second and third terms in \eqref{eq:seAn}. 
Let $\hat{\Sigma}_{jj}$ denote the estimator of $\Sigma_{jj}$ under OLS. 
By the property of OLS \citep{seber2012linear} and Condition \ref{cond:moments}, we have $\hat{\Sigma}_{jj}\to_p \Sigma_{jj}$ as $n\to \infty$ and 
\begin{align} \label{eq:sealphaj}
     \hat{se}^2(\AnjOLS) =  (\vT_n^{\top} \PXP \vT_n)^{-1} \hat{\Sigma}_{jj}.
\end{align} 
Then by Slutsky's theorem and Condition \ref{cond:moments},
\begin{align*}
     n \hat{se}^2(\AnjOLS) \toP 
      \Sigma_{jj}/\sigma^2_T
     \quad \text{and} \quad
    \frac{1}{{n} \hat{se}^2(\AnjOLS)} \toP \frac{\sigma^2_T}{ \Sigma_{jj}}, 
\end{align*} 
where $\Sigma_{jj}/\sigma^2_T>0$ is fixed. 
Therefore, the second and third terms in \eqref{eq:seAn} can be obtained. 

\paragraph{Proof of \eqref{eq:seBn}.}
We first show $\sqrt{n}(\BnjOLS - \Bnj) = O_p(1)$. 
Note $ \BnOLS$ is the regression coefficient vector of $\vM_n$ under the linear regression $\vY_n \sim \vM_n+ \mathbf{D}_M$. 
By Frisch–Waugh–Lovell Theorem, we know
\begin{align} 
     \BnOLS 
    = 
    &~  (\vM_n^{\top} \PDMP \vM_n)^{-1} \vM_n^{\top} \PDMP \vY_n =\Bn + (\vE_n^{\top} \PDMP \vE_n)^{-1} \vE_n^\top \PDMP \vepsilon_n, \label{eq:BnOLS}
\end{align}
where the second equation follows by the model $\vY_n = \vM_n \Bn + \mathbf{D}_M (\eta^*, \boldsymbol{\zeta}_{Y}^{*\top})^{\top} + \vepsilon_n$ by \eqref{eq:model_n}. 
To finish the proof, it suffices to show 
\begin{align}
   \vE_n^{\top} \PDMP \vE_n/n \toP \boldsymbol{\Sigma}\quad \quad \text{ and }\quad \quad \vE_n^\top \PDMP \vepsilon_n/\sqrt{n}\tod \Gsn(\vzero_p, \sigma^2\SigmaE),   \label{eq:conv_E_PDMP_E_n_0}
\end{align}
where the latter is $O_p(1)$. 
By the law of large numbers and the central limit theorem, we know $ \vE_n^\top  \vE_n /n \toP \boldsymbol{\Sigma}$ and $\vE_n^\top\vepsilon_n/\sqrt{n}\tod \Gsn(\vzero_p, \sigma^2\SigmaE)$. 
Thus, by $\PDMP=\mathrm{I}_{n\times n}-\PDM$ and Slutsky's theorem, it remains to show 
\begin{align}
    \frac{\vE_n^\top \PDM \vE_n  }{n}  
=&~ \frac{\vE_n^\top  \X (\X^{\top}\X)^{-1}\X^{\top} \vE_n}{n}  +  \frac{\vE_n^{\top}\mathbf{P}_{\X}^{\perp}\vT_n (\vT_n^{\top} \mathbf{P}_{\X}^{\perp}\vT_n)^{-1} \vT_n^{\top} \mathbf{P}_{\X}^{\perp} \vE_n}{n}   \notag   \\
\frac{\vE_n^\top \PDM \vepsilon_n}{\sqrt{n}}=&~  \frac{\vE_n^\top  \X (\X^{\top}\X)^{-1}\X^{\top} \vepsilon_n}{\sqrt{n}}  +  \frac{\vE_n^{\top}\mathbf{P}_{\X}^{\perp}\vT_n (\vT_n^{\top} \mathbf{P}_{\X}^{\perp}\vT_n)^{-1} \vT_n^{\top}\mathbf{P}_{\X}^{\perp}  \vepsilon_n}{\sqrt{n}}   \notag 
\end{align}
are $o_p(1)$, 
where the equations hold by  $ \PDM=   \mathbf{P}_{ (\vT_n, \mathbf{X}_n) }  = \mathbf{P}_{\X} + \mathbf{P}_{\mathbf{P}_{\X}^{\perp}\vT_n}$.  
In particular, the conclusion follows because by Condition \ref{cond:moments} and Markov's inequality, we have 
$\mathbf{X}_n^{\top}\mathbf{X}_n/n  \to \boldsymbol{\Sigma}_X$,  $\vT_n^{\top}\PXP\vT_n /n\to \sigma_T^2$ ,   $\mathbf{X}_n^{\top}\vE_n=O_p(\sqrt{n})$,
$\mathbf{X}_n^{\top}\vepsilon_n/n=O_p(\sqrt{n})$, $ \vT_n^{\top}\PXP \vE_{n}= O_p(\sqrt{n})$, and $ \vT_n^{\top}\PXP \vepsilon_{n}  = O_p(\sqrt{n})$. 

We next prove the second and third terms in \eqref{eq:seBn}. 
Let $\hat{\sigma}^2$ denote the estimator of $\sigma^2$ under OLS. By the properties of OLS and Condition \ref{cond:moments}, we have  $\hat{\sigma}^2\to_p \sigma^2$ as $n\to \infty$ and 
\begin{align} \label{eq:sebeta}
     \hat{se}^2(\BnjOLS) =  \left((\vE_n^{\top} \PDMP \vE_n)^{-1}\right)_{jj} \hat{\sigma}^2.
\end{align} 
Then by \eqref{eq:conv_E_PDMP_E_n_0} and \eqref{eq:sebeta}, 
\begin{align*}
     n\hat{se}^2(\BnjOLS) \toP \sigma^2{(\SigmaE^{-1})_{jj}} \quad \text{and} \quad
    \frac{1}{{n} \hat{se}^2(\BnjOLS)} \toP \frac{1}{\sigma^2{(\SigmaE^{-1})_{jj}}},
\end{align*}
where  $(\SigmaE^{-1})_{jj}$, the $j$th diagonal element of $\SigmaE^{-1}$, is positive and fixed by Condition \ref{cond:moments}. 
Therefore, the second and third terms in \eqref{eq:seBn} can be obtained. 

\section{Proof of Theorem \ref{thm:oracle_property}}
\label{sec:proof_oracle}

This section is organized as follows. 
Section \ref{subsec:prelemma} defines notation and lemmas to be used in the proof of Theorem \ref{thm:oracle_property}. 
Section \ref{subsec:proof_theorem} provides the main proof of Theorem \ref{thm:oracle_property}. 
Sections \ref{subsec:pflm1} and \ref{subsec:pflm2} prove Lemmas \ref{lemma:third_term} and \ref{lemma:asymptotics} given in Section \ref{subsec:prelemma}. 

\subsection{Notation and Preliminary Lemmas} \label{subsec:prelemma}

\paragraph{Notation.} 
For an index set, we use a superscript $c$ to denote its complement, and we use the index set itself as a subscript to a vector or matrix to represent its corresponding subvector or submatrix. For example, $\AS[c]$ represents the complement of $\AS\subseteq \{1, \ldots, p\}$, $\boldsymbol{\mathbf{\alpha}}_{\AS}^{\ast}$ 
represents the subvector of $\An$ consisting of the elements with indices in $\AS$, and $\SigmaE_{\AS}$ represents the submatrix of $\SigmaE$ consisting of the elements with both row and column indices in $\AS$. 
Additionally, we define 
\begin{equation*}
    \ASA = \left\{j: \Anj \neq 0,  \ \  j=1,\ldots, p \right\} \quad\text{and}\quad
    \ASB = \left\{j: \Bnj \neq 0,  \ \  j=1,\ldots, p \right\}
\end{equation*}
as the index sets of mediators with nonzero exposure-to-mediator effects and with nonzero mediator-to-outcome effects.

\begin{lemma}
\label{lemma:third_term}
    Under the conditions in Theorem \ref{thm:oracle_property}, as the sample size $n\to \infty$, 
    \begin{equation}
    \label{eq:convergence_penal}
    \frac{\lambdaA}{\sqrt{n}\wA} \toP 
    \begin{cases}
        0, &\text{for $j\in\ASA$}\\
        \infty, &\text{for $j\in\ASA[c]$}
    \end{cases}
    \quad \text{and} \quad
\frac{\lambdaB}{\sqrt{n}\wB} \toP
    \begin{cases}
        0, &\text{for $j\in\ASB$}\\
        \infty, &\text{for $j\in\ASB[c]$}
    \end{cases}.
    \end{equation}
\end{lemma}
\begin{proof}
    See Section \ref{subsec:pflm1}.
\end{proof}

\begin{lemma}
\label{lemma:asymptotics}
    Under the conditions in Theorem \ref{thm:oracle_property}, the \newm{} estimators of $\An$ satisfy:
    \begin{equation}
    \label{eq:asym_mediator}
        \sqrt{n}\left(\Anhat[,\ASA] - \AnASA\right) \tod \Gsn\left(\vzero, \sigma_T^{-2} \SigmaE_{\ASA}\right),
        \quad\text{and}\quad 
        \sqrt{n}\Anhat[{,\ASA[c]}] \toP \vzero,
    \end{equation}
    and the \newm{} estimators of $\Bn$ satisfy:
    \begin{equation}
    \label{eq:asym_outcome}
        \sqrt{n}\left(\Bnhat[,\ASB] - \BnASB\right) \tod \Gsn\left(\vzero, \sigma^2 \SigmaE_{\ASB}^{-1}\right)
        \quad\text{and}\quad 
        \sqrt{n}\Bnhat[{,\ASB[c]}]  \toP \vzero.
    \end{equation} 
\end{lemma}

\begin{proof}
    See Section \ref{subsec:pflm2}.
\end{proof}

\subsection{Proof of Theorem \ref{thm:oracle_property}} 
\label{subsec:proof_theorem}

First, to prove $\Fnorm{\Anhat-\An}^2 + \Fnorm{\Bnhat - \Bn}^2 \toP 0$, 
note that Lemma \ref{lemma:asymptotics} implies $\Anjhat - \Anj\toP 0$ and $\Bnjhat - \Bnj\toP 0$ for $j = 1, \ldots, p$. Then the conclusion follows from Slutsky's theorem. 

We next prove \eqref{eq:selection_consistency}. As
$\{\AShat_n = \AS\}^c=   \{\AShat_n \supseteq \AS\}^c\cup  \{\AShat_n \subseteq \AS\}^c,$ by Boole's inequality, to prove \eqref{eq:selection_consistency}, it suffices to show
\begin{align}
   \lim_{n\to \infty} \Pr( \{\AShat_n \supseteq \AS\}^c) = 0, \label{eq:selection_concsistency_included}\\
   \lim_{n\to \infty} \Pr( \{\AShat_n \subseteq \AS\}^c ) = 0. \label{eq:target3_1}
\end{align} 

\paragraph{(i) Proof of \eqref{eq:selection_concsistency_included}.} 
Note 
\begin{align*}
    \Pr(\{\AShat_n \supseteq \AS\}^c)=&~ \Pr( \cup_{j\in \AS}\{ j \not \in \AShat_n \} )  
    \leqslant \sum_{j \in \AS} \Pr (j \not \in \AShat_n) \\
   \leqslant &~ \sum_{j \in \AS}\{\Pr(\Anjhat=0)+\Pr(\Bnjhat=0) \}, 
\end{align*}
where the last inequality follows as $j\not\in \AShat_n$ is equivalent to 
$\{\Anjhat=0\}\cup \{\Bnjhat=0\}$ based on our construction. 
Note
\begin{align}
    \Pr(\Bnjhat=0) = \Pr\left(\Bnj = - (\Bnjhat - \Bnj) \right) 
    \leq \Pr\left(\sqrt{n}\left|\Bnj\right| = \sqrt{n}\left|\Bnjhat - \Bnj \right| \right). \label{eq:pr_j_ASBc}  
\end{align}
By the first part of \eqref{eq:asym_outcome}, $\eqref{eq:pr_j_ASBc} \to 0$ as $n \to \infty$ for any $j \in\AS$.  Similarly, we can use \eqref{eq:asym_mediator} to show that $\lim_{n\to\infty}\Pr(\Anjhat =0) = 0$ for all $j\in \AS$. In summary,  $\lim_{n\to \infty}  \Pr(\{\AShat_n \supseteq \AS\}^c) =0$, and \eqref{eq:selection_concsistency_included} is proved.

\paragraph{(ii) Proof of \eqref{eq:target3_1}.}
Since $  \{\AShat_n \subseteq \AS\}^c=\{\text{There exists } j\not \in \AS \text{ such that } j\in \AShat_n\} = \cup_{j \not\in \AS} \{j\in \AShat_n\},$ and $\{j\in  \AShat_n\}=\{\Anjhat \neq 0 \}\cap \{\Bnjhat\neq 0\}$ by our construction, 
we have 
\begin{align}
    \Pr( \{\AShat_n \subseteq \AS\}^c )\leqslant &~ \sum_{j \not\in \AS }\Pr(\{\Anjhat \neq 0 \}\cap \{\Bnjhat\neq 0\} )\notag \\
    = &~  \sum_{j \in (\ASA\cap \ASB)^c} \Pr(\{\Anjhat \neq 0 \}\cap \{\Bnjhat\neq 0\} ) \quad \quad \text{(by } \AS=\ASA\cap \ASB) \notag\\
    \leqslant &~ \sum_{j \in \ASA[c]}\Pr(\Anjhat \neq 0) + \sum_{j \in \ASB[c] } \Pr(\Bnjhat\neq 0). \label{eq:sumprobtwo}
\end{align}

First, $\Anjhat \neq 0 $, by Remark \ref{rm:alpha_sol}, is equivalent to 
\begin{equation*}
    \frac{|\vM_{nj}^\mytrans \mathbf{P}_{\vect{X}_n}^{\perp}\vT_n|} {\|\mathbf{P}_{\vect{X}_n}^{\perp}\vT_n\|_2^2} >  \frac{\lambdaA}{2 \wA \|\mathbf{P}_{\vect{X}_n}^{\perp}\vT_n\|_2^2}.
\end{equation*}
Multiplying both sides of the above inequality by $\|\PXP \vT_n\|_2^2/\sqrt{n}$ and substituting the model of $\vM_{nj}$ in  \eqref{eq:model_n}, we obtain
\begin{equation}
\label{eq:ineq_ASAc}
    \left|\frac{\|\PXP \vT_n\|_2^2}{\sqrt{n}} \Anj + \frac{\vE_{nj}\PXP \vT_n}{\sqrt{n}}
    \right|> \frac{1}{2} \frac{\lambdaA}{\sqrt{n} \wA}.
\end{equation}
For $j \in \ASA[c]$, $\Anj=0$, thus the left hand side of \eqref{eq:ineq_ASAc} is $O_p(1)$ by \eqref{eq:conv_E_PXP_T_rootn}, while the right hand side of \eqref{eq:ineq_ASAc} $\toP \infty$ as $n\to\infty$ by \eqref{eq:convergence_penal}.

Second, by the Karush–Kuhn–Tucker (KKT) condition and \eqref{eq:opt}, we know the estimate $\Bnhat$ satisfies
\begin{align*}
    -2\DM^\top \PDMP(\vY_n-\vM_n\Bnhat) + \lambdaB \sum_{j=1}^p \frac 1 {\wB} \left.\frac{\partial |\beta|}{\partial\beta}\right|_{\beta = \Bnjhat} = 0.
\end{align*}
For $\Bnjhat\neq 0$, it reduces to 
\begin{equation*}
    2\vM_{nj}^\mytrans \PDMP (\vY_n - \vM_n \Bnhat)= \frac{\lambdaB}{\wB}\sgn(\Bnjhat).
\end{equation*}
Dividing both sides of the equation by $2\sqrt{n}$ and substituting $\vY_n$ and $\vM_n$ using \eqref{eq:model_n}, we obtain
\begin{equation}
\label{eq:KKT}
    \frac{\vE_{nj}^\mytrans \PDMP\vect{\epsilon}_n}{\sqrt{n}} + \frac{\vE_{nj}^\mytrans \PDMP\vE_n}{n} \sqrt{n} (\Bn - \Bnhat) = \frac{1}{2}\frac{\lambdaB}{\sqrt{n}\wB}\sgn(\Bnjhat).
\end{equation}
For any $j \in\ASB[c]$,  the left hand side of \eqref{eq:KKT} is $O_p(1)$ by \eqref{eq:conv_E_PDMP_E_n_0} and the first part of \eqref{eq:asym_outcome}, while the absolute value of the right hand side of \eqref{eq:KKT} $\toP \infty$ as $n\to\infty$ by \eqref{eq:convergence_penal}. Therefore, for any $j \in\ASB[c]$,
\begin{align*}
    \Pr\left(j \in \ASBhat\right)
    & \leq \Pr \left(\left| \frac{\vE_{nj}^\mytrans \PDMP\vect{\epsilon}_n}{\sqrt{n}} + \frac{\vE_{nj}^\mytrans \PDMP\vE_n}{n} \sqrt{n} (\Bn - \Bnhat)\right| = \frac{1}{2}\frac{\lambdaB}{\sqrt{n}\wB}
    \right) \\
     &\to 0 \quad\text{as}\quad n \to \infty.
\end{align*}

In summary, we obtain $ \eqref{eq:sumprobtwo} \to 0$ as $n\to \infty$, which finishes the proof. 

\subsection{Proof of Lemma \ref{lemma:third_term}}
\label{subsec:pflm1}

First, to prove $\frac{\lambdaA}{\sqrt{n}\wA} \toP \infty$ for $j \in \ASA[c]$, since $\lambdaA \gg n^{1/2-\gATwo}$, it suffices to show $n^{\gATwo}\wA = O_p(1)$ for $j\in\ASA[c]$. 
Note that for $j\in\ASA[c]$, $\sqrt{n}\AnjhatInit = \sqrt{n}(\AnjhatInit - \Anj) = O_p(1)$ by Condition \ref{cond:initial}(i), then
\begin{align*}
    n^{\gATwo}\wA 
    &= n^{\frac{2\gATwo - \gAOne}{2}} |\sqrt{n}\AnjhatInit|^{\gAOne}|\BnjhatInit|^{\gAOne} + |\sqrt{n}\AnjhatInit|^{2\gATwo} = O_p(1),
\end{align*}
which follows by $\gAOne > 2\gATwo$.

Second, we show $\frac{\lambdaA}{\sqrt{n}\wA} \toP 0$ for $j \in \ASA$. Following Condition \ref{cond:initial}(i) and Slutsky's theorem, as $n\to \infty$,
\begin{equation*}
    \wA = |\AnjhatInit\BnjhatInit|^{\gAOne} + |\AnjhatInit|^{2\gATwo} \toP |\Anj\Bnj|^{\gAOne} + |\Anj|^{2\gATwo} > 0\quad\text{for $j\in\ASA$}.
\end{equation*}
Since $\lambdaA \ll n^{1/2}$, $\frac{\lambdaA}{\sqrt{n}\wA} \toP 0$ for $j \in \ASA$.

Following similar arguments, we can prove the conclusions for $\frac{\lambdaB}{\sqrt{n}\wB}$.

\subsection{Proof of Lemma \ref{lemma:asymptotics}}
\label{subsec:pflm2}

\paragraph{(i) Proof of \eqref{eq:asym_outcome}.}
To prove \eqref{eq:asym_outcome}, we use the formula for penalized coefficients in Proposition \ref{proposition:solution}.
In particular, under the loss function in \eqref{eq:quadloss_theory}, 
we have $\boldsymbol{\theta}_{YP} = \vbeta$,  $\bar{\mathcal{P}}_Y(\vbeta) =  \lambdaB \sum_{j = 1}^p {|\beta_{j}|}/{\wB}$, and 
\begin{align*}
    \ell_{Y,1}(\vbeta) =&~ \Fnorm{\PDMP(\vY_n - \vM_{n} \vbeta)}^2\\
    =&~ \Fnorm{ \PDMP\{\vM_n(\vbeta^*-\vbeta)+ \vepsilon_n\} }^2 =  \Fnorm{ \PDMP\vE_n(\vbeta^*-\vbeta)+ \PDMP\vepsilon_n }^2,
\end{align*}
where the second and third equations follow by the model of $\mathbf{Y}_n$ and $\mathbf{M}_n$ in \eqref{eq:model_n}. 

Define $\vunhatB = \sqrt{n}(\Bnhat - \Bn)$. 
By \eqref{eq:thetahatp_pf}, we have $\vunhatB = \argmin_{\vu \in \reals^p} \VnB \left(\vu\right)$, where we define 
\begin{align}
    \VnB(\vu) 
    \equiv &~ \left\{\ell_{Y,1}(\Bn + \frac{\vu}{\sqrt{n}}) + 
    \bar{\mathcal{P}}_Y(\Bn + \frac{\vu}{\sqrt{n}})\right\} - \{\ell_{Y,1}(\Bn) + \bar{\mathcal{P}}_Y(\Bn)\}\notag\\
    =&~ \vu^\mytrans\left(\frac {\vE_n^{\top} \PDMP \vE_n}{n}\right) \vu
        -2\vu^{\top} 
        \frac{\vE_n^{\top}\PDMP\vepsilon_n}{\sqrt{n}}\notag\\
    &+  \sum_{j=1}^p \frac{\lambdaB}{\sqrt{n}\wB} \sqrt{n}\left(\left|\Bnj + \frac{u_j}{\sqrt{n}}\right| - \left|\Bnj\right|\right), \label{eq:VnBvu}
\end{align}
which follows by plugging the formulae of $\ell_{Y,1}(\cdot)$ and $\bar{\mathcal{P}}_Y(\cdot)$. 
To obtain the asymptotics of $\vunhatB$, we derive the limit of $\VnB(\cdot)$.

First, in  $\VnB(\vu)$, $ \frac{\lambdaB}{\sqrt{n}\wB} \sqrt{n}\left(\left|\Bnj + \frac{u_j}{\sqrt{n}}\right| - \left|\Bnj\right|\right)=0$ if $u_j=0$. 
When $u_j \neq0$, its convergence can be discussed in two cases. For any $j \in \ASB$ and fixed $u_j\neq 0$, by \eqref{eq:convergence_penal}, 
\begin{equation}
\label{eq:Vn_third_A}
    \left|\frac{\lambdaB}{\sqrt{n}\wB} \sqrt{n}\left(\left|\Bnj + \frac{u_j}{\sqrt{n}}\right| - \left|\Bnj\right|\right)\right| \leq \frac{\lambdaB}{\sqrt{n}\wB} |u_j| \toP 0 \quad \text{as} \quad n\to\infty.
\end{equation}
For any $j \in \ASB[c]$ and fixed $u_j\neq 0$,
\begin{equation}
\label{eq:Vn_third_Ac}
    \frac{\lambdaB}{\sqrt{n}\wB} \sqrt{n}\left(\left|\Bnj + \frac{u_j}{\sqrt{n}}\right| - \left|\Bnj\right|\right) = \frac{\lambdaB}{\sqrt{n}\wB} |u_j|\toP \infty \quad \text{as} \quad n\to\infty.
\end{equation}
By \eqref{eq:conv_E_PDMP_E_n_0} and Slutsky's theorem, we have that for any fixed $\vu$, as $n\to\infty$,
\begin{align} 
 \VnB(\vu) \tod \VB(\vect{u}) \equiv &~ 
    \begin{cases}
        \vu^\mytrans \SigmaE\vu - 2 \vu^\mytrans \vectrv{W} &  \hspace{0.5em}\text{if $u_j = 0 \text{ for all } j \in \ASB[c]$,} \\
        \infty & \hspace{0.5em} \text{otherwise,}
    \end{cases} \label{eq:VnBLimit}\\
    =&~  \begin{cases}
        \vu_{\ASB}^\mytrans \SigmaE_{\ASB}\vu_{\ASB} - 2 \vu_{\ASB}^\mytrans \vectrv{W}_{\ASB}  & \text{if $u_j = 0 \text{ for all } j \in \ASB[c]$,} \\
        \infty & \text{otherwise,}
    \end{cases}  \notag
\end{align} 
where $\vectrv{W} \sim \Gsn\left(\vzero, \sigma^2 \SigmaE\right)$,  
and $\vectrv{W}_{\ASB} $ denotes  its subvector satisfying $\vectrv{W}_{\ASB} \sim \Gsn(\vzero, \sigma^2 \SigmaE_{\ASB})$. 
Since $\SigmaE_{\ASB}$ is positive definite, $\VB(\vect{u})$ has a unique minimizer. 
Following the epi-convergence results in \citet{geyer1994asymptotics}, we have
\begin{align}\label{eq:asym_u_outcome}
   \vunhatB = \argmin_{\vu} \VnB(\vu) \tod \argmin_{\vu} \VB(\vu) \ \quad \Rightarrow \quad \ \begin{cases}
   \hat{\vu}_{n, \ASB}^{\beta} \tod \SigmaE_{\ASB}^{-1} \vectrv{W}_{\ASB},\\
   \hat{\vu}_{n,\ASB[c]}^{\beta} \tod \vzero.
   \end{cases}
\end{align}
Therefore, \eqref{eq:asym_outcome} is proved.

\paragraph{(ii) Proof of \eqref{eq:asym_mediator}.}
The proof is similar to the proof of \eqref{eq:asym_outcome}. 
To prove \eqref{eq:asym_mediator}, we use the simplified formula for penalized coefficients in Proposition \ref{proposition:solution}.
Specifically, under the loss function in \eqref{eq:quadloss} with $\mathcal{P}_M(\cdot) = 0$, we have $\boldsymbol{\theta}_{MP} = \valpha$,   $\bar{\mathcal{P}}_M(\valpha) =  \lambdaA \sum_{j = 1}^p {|\alpha_{j}|}/{\wA}$, and
\begin{align*}
    \ell_{M,1}(\valpha) = \Fnorm{\PXP(\vM_n - \vT_{n} \valpha^\top)}^2 = \Fnorm{\PXP \vE_n +\PXP\vT_n(\valpha^{*\top} - \valpha^{\top}) }^2,
\end{align*}
where the second equation follows by the model of $\vM_n$ in \eqref{eq:model_n}. 

Define $\vunhatA = \sqrt{n}(\Anhat - \An)$. 
By \eqref{eq:thetahatp_pf}, we have  $\vunhatA = \argmin_{\vu \in \reals^p} \VnA \left(\vu\right)$, where we define
\begin{align}
    \VnA(\vu) 
    \equiv & \left(\ell_{M,1}(\An + \frac{\vu}{\sqrt{n}}) + 
    \bar{\mathcal{P}}_M(\An + \frac{\vu}{\sqrt{n}})\right) - (\ell_{M,1}(\An) + \bar{\mathcal{P}}_M(\An))\notag\\
    =& \frac{ \|\PXP \vT_n\|_2^2}{n}\vu^{\top} \vu 
        - 2 \vu^{\top} \frac{\vE_n^{\top}\PXP \vT_n}{\sqrt{n}}
    + \sum_{j=1}^p  \frac{\lambdaA}{\sqrt{n}\wA} \sqrt{n}\left(\left|\Anj + \frac{u_j}{\sqrt{n}}\right| - \left|\Anj\right|\right),\label{eq:VnAvu}
\end{align}
which follows by plugging the formulae of $\ell_{M,1}(\cdot)$ and $\bar{\mathcal{P}}_M(\cdot)$. 
We derive the limit of $\VnA(\cdot)$ in \eqref{eq:VnAvu} to obtain the asymptotics in \eqref{eq:asym_mediator}.

For the third term in \eqref{eq:VnAvu}, 
following analogous reasoning to the discussion above, we can show that for any fixed $u_j\neq 0$, as $n\to\infty$,
\begin{align*}
     \frac{\lambdaA}{\sqrt{n}\wA} \sqrt{n}\left(\left|\Anj + \frac{u_j}{\sqrt{n}}\right| - \left|\Anj\right|\right)\toP
     \begin{cases}
         0 \quad &\text{if}\quad j\in\AS,\\
         \infty \quad &\text{if}\quad j\in\AS[c].
     \end{cases}
\end{align*}
Combining with Condition \ref{cond:moments} and \eqref{eq:conv_E_PXP_T_rootn}, by Slutsky's theorem, we have that for any fixed $\vu$, as $n\to\infty$,
\begin{align}
     \VnA(\vu) \tod  \VA(\vu) \equiv &~ 
     \begin{cases}
         \sigma_T^2 \vu^\top\vu - 2 \vu^\top \vectrv{W}' & \hspace{3.2em}\text{if $u_j = 0 \text{ for all } j \in \ASA[c]$,} \\
        \infty & \hspace{3.2em} \text{otherwise,}
     \end{cases}\label{eq:VnALimit}\\
      =&~  \begin{cases}
        \sigma_T^2 \vu_{\ASA}^\mytrans\vu_{\ASA} - 2 \vu_{\ASA}^\mytrans \vectrv{W}_{\ASA}' & \ \text{if $u_j = 0 \text{ for all } j \in \ASA[c]$,} \\
        \infty & \ \text{otherwise,}
    \end{cases}\notag
\end{align}
where $\vectrv{W}' \sim \Gsn\left(\vzero, \sigma_T^2 \SigmaE\right)$, and $\vectrv{W}'_{\ASA}$ denotes its subvector satisfying $\vectrv{W}'_{\ASA} \sim \Gsn\left(\vzero, \sigma_T^2 \SigmaE_{\ASA}\right)$. 
Since $\SigmaE_{\ASA}$ is positive definite, $\VA(\vect{u})$ has a unique minimizer. 
Following the epi-convergence results in \citet{geyer1994asymptotics}, we have
\begin{align}\label{eq:asym_u_mediator}
   \vunhatA = \argmin_{\vu} \VnA(\vu) \tod \argmin_{\vu} \VA(\vu) \ \quad \Rightarrow \quad \ \begin{cases}
        \hat{\vu}_{n, \ASA}^{\alpha} \tod \sigma_T^{-2} \vectrv{W}'_{\ASA},\\
        \hat{\vu}_{n,\ASA[c]}^{\alpha} \tod \vzero.
   \end{cases}
\end{align}
Therefore, \eqref{eq:asym_mediator} is proved. 

\section{Proof of Proposition prop:penaltyscale}\label{sec:pf_penaltyscale}

We have that
\begin{equation}
    \frac{\hat{w}_{nj, \alpha}}{\hat{w}_{\mathrm{AL}, nj, \alpha}} = 1 + |\AnjhatInit|^{\gAOne - 2 \gATwo} |\BnjhatInit|^{2 \gAOne} 
    \quad \text{and} \quad 
    \frac{\hat{w}_{nj, \beta}}{\hat{w}_{\mathrm{AL}, nj, \beta}} = 1 + |\AnjhatInit|^{2 \gBOne} |\BnjhatInit|^{\gBOne - 2 \gBTwo}.
\end{equation}
When $0 < 2 \gATwo < \gAOne$ and $0 < 2 \gBTwo < \gBOne$, \eqref{eq:weight_ratio} follows from Condition 1 and Slutsky's Theorem.

\section{Pathway LASSO}\label{sec:PathwayLASSO}

In this section, we present supplementary simulation studies on Pathway LASSO. 

\emph{Setup.} We adopt a data generation mechanism that is a special case of that in Section \ref{sec:simulations} with a focus on Case (I) of $\vectrv{E}$ with $\rho = 0$. 
We set $n=2000$ and $\BT^\ast = 1$. 
For the pairwise coefficients $\{(\alpha_{j}^*,\beta_j^*): j=1,\ldots, p\}$, we still consider the six groups of patterns as defined in \eqref{eq:coeffgroups} and Table \ref{table:sim_setting} with $\bSim = 2^{-1.5}$. Under each setting, we simulate $100$ independent datasets.

\begin{figure}
    \centering
    \includegraphics[width=\textwidth]{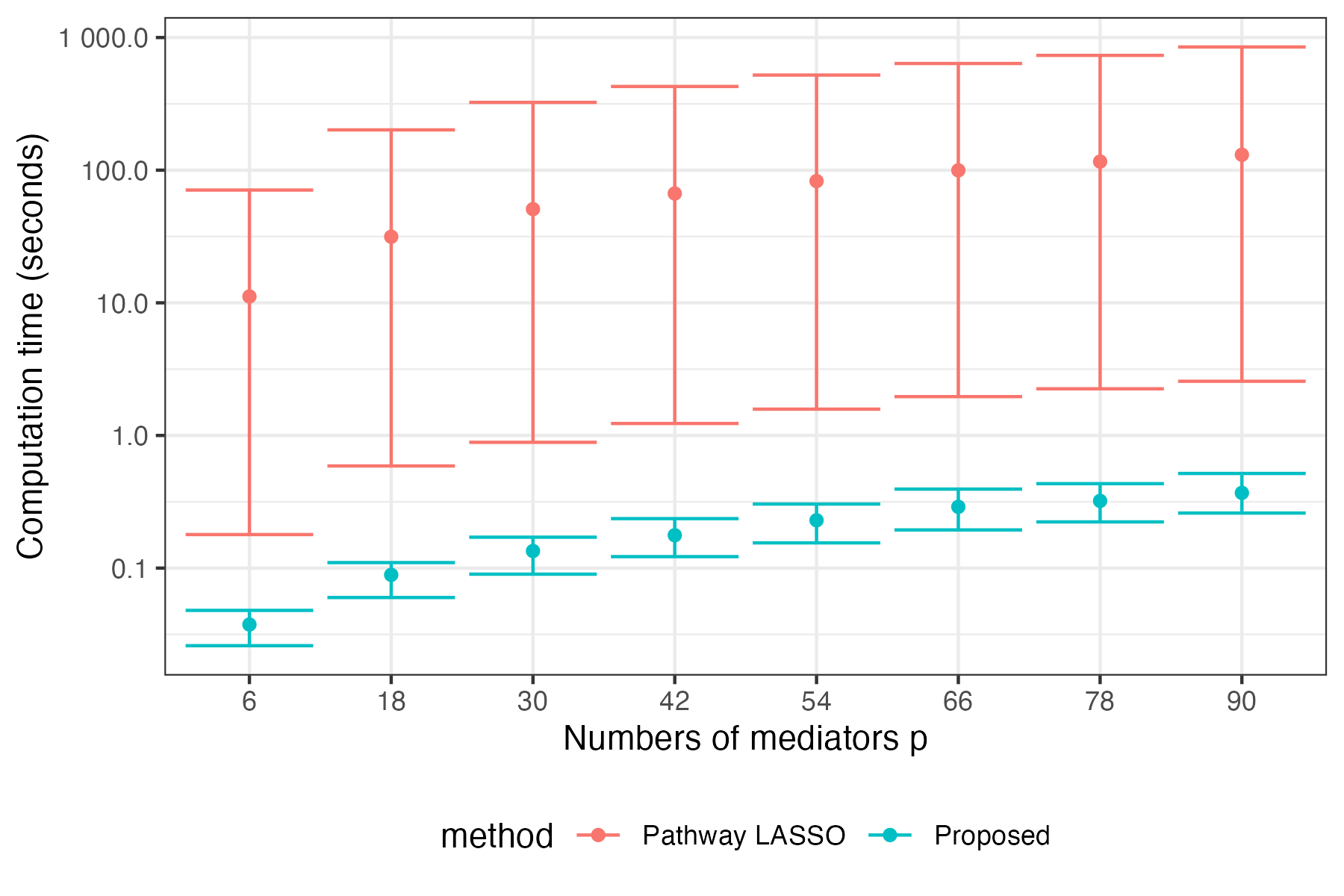}
    \caption{Empirical computation time  for fitting Pathway LASSO and \newm{} once across $100$ repetitions. The dots represent the average computation time, and the error bars indicate the $0.05$ and $0.95$ quantiles.}
    \label{fig:Pathway_runtime}
\end{figure}

\medskip 
\emph{Computation time.} 
The computation time of Pathway LASSO grows rapidly in $p$, making it prohibitive for the $p = 150$ setting considered in Section \ref{sec:simulations}. To illustrate, we next consider a smaller range of $p \in \{ 6, 18, 30, 42, 54, 66, 78, 90\}$. 
Figure \ref{fig:Pathway_runtime} compares the empirical computation time between Pathway LASSO and \newm{} when they are fitted once.   
Notably, Pathway LASSO is almost always $300$ times more time-consuming than \newm{}. 
 For our simulations in Section \ref{sec:simulations}, 
including Pathway LASSO is even more prohibitive as a larger $p=150$ is used, and cross-validation is applied for optimal penalty parameter $\lambda$.   
As an example, when $p = 90$, selecting the optimal $\lambda$ from $40$ candidates using $5$-fold cross-validation takes Pathway LASSO over $14$ hours, whereas \newm{} completes the same task in under $3$ minutes. 

\medskip 
\emph{Selection accuracy.} 
Due to the computational burden discussed above, we evaluate the selection accuracy of Pathway LASSO for  $p\in \{6,\ldots, 90\}$, which are smaller than $p=150$ in Section \ref{sec:simulations}, and their default range of hyperparameters is considered. 
Specifically, their penalty hyperparameter $\lambda$ is chosen from values 
ranging from $0.0014$ to $100000$, formed by combining two uniformly logarithmically spaced sequences, featuring a denser grid for smaller magnitudes and a sparser grid for larger magnitudes. 
When fitting JAP for comparison, we let $\lambdaA$ and $\lambdaB$ in \eqref{eq:fit_mediator} and \eqref{eq:fit_outcome} take the same range of forty values as $\lambda$ in Pathway LASSO,  
while fixing $\gAOne = \gATwo = \gATwo = \gBTwo = 1$ in \eqref{eq:wweights}.
Across all $p$ values, 
the empirical accuracy of selecting $\AS$ of JAP under the tuned hyperparameters is 1, and that of Pathway LASSO under all values of $\lambda$ are 0. 
In future studies, 
it could be of interest to improve the accuracy of Pathway LASSO  by exploring a larger range of candidate hyperparameters.

\end{document}